\documentclass[11pt,a4paper,final]{article}
\usepackage{amsmath}%
\usepackage{amstext}%
\usepackage{amssymb}%
\usepackage{showkeys}%
\usepackage{epsfig}%
\usepackage{graphicx}%
\usepackage{multicol}
\usepackage{xcolor}
\usepackage{apacite}
\usepackage{amsthm}
\usepackage{actuarialangle}
\usepackage{lifecon}
\usepackage{url}
\usepackage[top=1in, bottom=1.25in, left=0.8in, right=0.8in]{geometry}

\theoremstyle{prop}
\theoremstyle{proof}

\newtheorem{lem}{Lemma}%[section]
\begin{document}

\title{The Log Private Company Valuation Model}
\author{Battulga Gankhuu\footnote{
Department of Applied Mathematics, National University of Mongolia, Ulaanbaatar, Mongolia;
E-mail: battulga.g@seas.num.edu.mn}}
\date{}


\maketitle

\begin{abstract}
For a public company, pricing and hedging models of options and equity--linked life insurance products have been sufficiently developed. However, for a private company, because of unobserved prices, pricing and hedging models of the European options and life insurance products are in their early stages of development. For this reason, this paper introduces a log private company valuation model, which is based on the dynamic Gordon growth model. In this paper, we obtain closed--form pricing and hedging formulas of the European options and equity--linked life insurance products for private companies. Also, the paper provides Maximum Likelihood (ML) estimations of our model, Expectation Maximization (EM) algorithm, and valuation formula for private companies. \\[3ex]


%\textbf{GEL:} C32, G12, G22, G32.
\end{abstract}

\section{Introduction}

Dividend discount models (DDMs), first introduced by \citeA{Williams38}, are popular tools for stock valuation. The basic idea is that the market price of a stock is equal to the present value of the sum of the dividend paid by the firm and the stock price of the firm, corresponding to the next period. As the outcome of DDMs depends crucially on dividend forecasts, most research in the last few decades has been around the proper model of dividend dynamics. Reviews of some existing DDMs that include deterministic and stochastic models can be found in \citeA{dAmico20a} and \citeA{Battulga22a}.

Existing stochastic DDMs have one common disadvantage: If dividends have chances to take negative values, then the stock price of the firm can take negative values with a positive probability, which is an undesirable property for a stock price. A log version of the stochastic DDM, which is called by dynamic Gordon growth model was introduced by \citeA{Campbell88}, who derived a connection between log price, log dividend, and log return by approximation. Since their model is in a log framework, the stock price and dividend get positive values.

A private company is a firm held under private ownership. Private companies may issue stock and have shareholders, but their shares do not trade on public exchanges and are not issued through an initial public offering (IPO). As a result, the stock prices of private companies are unobserved most of the time. On the other hand, a public company is a company that has sold a portion of itself to the public via an IPO, meaning shareholders have a claim to part of the company's assets and profits. Consequently, the public knows the daily prices of public companies.

Black--Scholes European call and put options are contracts that give their owner the right, but not the obligation, to buy or sell shares of a stock of a company at a predetermined price by a specified date. For a public company, the European option pricing models have been developed. However, since the stock prices of a private company are unobserved most of the time, the pricing of the European call and put options was a difficult topic for the private company. For the private company, to motivate employees to adopt a personal investment in the company's success, frequently offered the call option, giving the holder the right to purchase shares of stock at a specified price of the company. Also, to motivate the employee the private company may offer equity--linked life insurance products to the workers. 

In this paper, we aim to contribute the following fields: (i) to price the European options, (ii) to connect the equity--linked life insurance products with the private company, (iii) to hedge the options and life insurance products, and (iv) to estimate the model's parameter and to value the private company. To estimate the required rate of return of a private company, \citeA{Battulga23b} considered the (iv)th field. But the (i)--(iii) fields still have not been explored before. Therefore, the novelty of the paper is threefold.

The rest of the paper is organized as follows: In Section 2, we introduce a log private company valuation model. In Section 3, we obtain closed--form Black--Scholes call and put options pricing formula. Section 4 provides connections between the log private company valuation model and equity--linked life insurance products. In Section 5, we consider locally risk--minimizing strategies for the options and equity--linked life insurance products. Section 6 is dedicated to ML estimators of the model's parameters and an estimator of the value of a private company. In Section 7, we conclude the study. Finally, in the Technical Annex, we provide Lemmas, which are useful for pricing and hedging and their proofs.

\section{Log Private Company Valuation Model}

Let $(\Omega,\mathcal{G}_T^x,\mathbb{P})$ be a complete probability space, where $\mathbb{P}$ is a given physical or real--world probability measure and $\mathcal{G}_T^x$ will be defined below. Dividend discount models (DDMs), first introduced by \citeA{Williams38}, are a popular tool for stock valuation. In this paper, we assume that a firm will not default in the future. For a DDM with default risk, we refer to \citeA{Battulga22a}. The basic idea of all DDMs is that the market price at time $t-1$ of stock of the firm equals the sum of the stock price at time $t$ and dividend at time $t$ discounted at risk--adjusted rates (required rate of return on stock). We suppose that there are $n$ private companies. Therefore, for successive prices of $i$--th company, the following relation holds 
\begin{equation}\label{03001}
P_{i,t}=(1+k_{i,t})P_{i,t-1}-d_{i,t}
\end{equation}
for $t=1,\dots,T$ and $i=1,\dots,n$, where $k_{i,t}$ is the required rate of return, $P_{i,t}$ is the stock price, and $d_{i,t}$ is the dividend, respectively, at time $t$ of $i$--th company. Let $P_t:=(P_{1,t},\dots,P_{n,t})'$ be an $(n\times 1)$ price vector at time $t$, $d_t:=(d_{1,t},\dots,d_{n,t})'$ be an $(n\times 1)$ dividend vector at time $t$, $k_t:=(k_{1,t},\dots,k_{n,t})'$ be an $(n\times 1)$ required rate of return vector at time $t$. Then, equation \eqref{03001} is compactly written by
\begin{equation}\label{03002}
P_t=(i_n+k_t)\odot P_{t-1}-d_t
\end{equation}
where $i_n:=(1,\dots,1)'$ is an $(n\times 1)$ vector, which consists of ones and $\odot$ is the Hadamard's element--wise product of two vectors.

As mentioned above, if dividends have chances to take negative values, then the stock prices of the firms can take negative values with positive probabilities, which is an undesirable property for the stock prices. For this reason, we follow the idea in \citeA{Campbell88}. As a result, the stock prices of the companies take positive values. Following the ideas in \citeA{Campbell88}, one can obtain the following approximation
\begin{equation}\label{03003}
\ln\big(i_n+\exp\{\tilde{d}_t-\tilde{P}_t\}\big)\approx \ln(g_t)+G_t^{-1}(G_t-I_n)\big(\tilde{d}_t-\tilde{P}_t-\mu_t\big)
\end{equation}
where for a generic function $f:\mathbb{R}\to\mathbb{R}$ and $(n\times 1)$ vector $o=(o_1,\dots,o_n)'$, we use a vector notation $f(o):=(f(o_1),\dots,f(o_n))'$, $I_n$ is an $(n\times n)$ identity matrix, $\tilde{P}_t:=\ln(P_t)$ is an $(n\times 1)$ log price vector at time $t$, $\tilde{d}_t:=\ln(d_t)$ is an $(n\times 1)$ log dividend vector at time $t$, $g_t:=i_n+\exp\{\mu_t\}$ is an $(n\times 1)$ linearization parameter vector at time $t$, $G_t:=\text{diag}\{g_t\}$ is an $(n\times n)$ diagonal matrix, whose diagonal elements consist of the linearization parameter vector $g_t$, and $\mu_{t}:=\mathbb{E}\big[\tilde{d}_{t}-\tilde{P}_{t}\big|\mathcal{F}_0\big]$ is an $(n\times 1)$ mean log dividend--to--price ratio process at time $t$ of the companies, where $\mathcal{F}_0$ is an initial information, see below. It should be noted that a vector, which consists of diagonal elements of the matrix $G_t^{-1}$ has the following explanation: ignoring the expectation, it represents a percentage of the stock price at time $t$ in a sum of the stock price at time $t$ and dividend payment at time $t$. Therefore, we have
\begin{equation}\label{03004}
\exp\{\tilde{k}_t\}=(P_t+d_t)\oslash P_{t-1}\approx \exp\big\{\tilde{P}_t-\tilde{P}_{t-1}+\ln(g_t)+G_t^{-1}(G_t-I_n)\big(\tilde{d}_t-\tilde{P}_t-\mu_t\big)\big\},
\end{equation}
where $\oslash$ denotes element--wise division of two vectors, $\tilde{k}_t:=\ln(1+k_t)$ is an $(n\times 1)$ log required rate of return vector. As a result, for the log price at time $t$, the following approximation holds
\begin{equation}\label{03005}
\tilde{P}_t\approx G_t(\tilde{P}_{t-1}-\tilde{d}_t+\tilde{k}_t)+\tilde{d}_t-h_t.
\end{equation}
where $h_{t}:=G_t\big(\ln(g_t)-\mu_t\big)+\mu_t$ is an $(n\times 1)$ linearization parameter and the model is referred to as dynamic Gordon growth model, see \citeA{Campbell88}. For a quality of the approximation, we refer to \citeA{Campbell97}. To estimate the parameters of the dynamic Gordon growth model, to price Black--Scholes call and put options and equity--linked life insurance products, which will appear in the following sections, and to hedge the options and life insurance products we assume that the log required rate of return process is a sum of deterministic part and random part, namely, $\tilde{k}_t=C_k\psi_t+u_t$ under the real probability measure $\mathbb{P}$, where $C_k$ is an ($n\times l$) matrix, $\psi_t:=(\psi_{1,t},\dots,\psi_{l,t})'$ is an $(l\times 1)$ vector of exogenous variables, and $u_t:=(u_{1,t},\dots,u_{n,t})'$ is an $(n\times 1)$ white noise process. In this case, equation \eqref{03005} becomes 
\begin{equation}\label{03006}
\tilde{P}_t=G_t(\tilde{P}_{t-1}-\tilde{d}_t+C_k\psi_t)+\tilde{d}_t-h_t+G_tu_t.
\end{equation}

Let $\mathsf{B}_{i,t}$ be a book value of equity and $b_{i,t}$ be a book value growth rate, respectively, at time $t$ of $i$--th company. Since the book value of equity at time $t-1$ grows at rate $b_{i,t}$, a log book value at time $t$ of $i$--th company becomes 
\begin{equation}\label{03007}
\ln(\mathsf{B}_{i,t})=\tilde{b}_{i,t}+\ln(\mathsf{B}_{i,t-1}),
\end{equation}
where $\tilde{b}_{i,t}:=\ln(1+b_{i,t})$ is a log book value growth rate of $i$--th company. On the other hand, let $\Delta_{i,t}:=d_{i,t}/B_{i,t-1}$ be a dividend--to--book ratio process at time $t$ of the $i$--th company. Then, we have
\begin{equation}\label{03008}
\tilde{d}_{i,t}=\tilde{\Delta}_{i,t}+\ln(\mathsf{B}_{i,t-1}),
\end{equation}
where $\tilde{\Delta}_{i,t}:=\ln(\Delta_{i,t})$ is a log dividend--to--book ratio process at time $t$ of $i$--th company. To obtain the dynamic Gordon growth model, corresponding to the private companies, let $\tilde{m}_{i,t}:=\ln(P_{i,t}/B_{i,t})$ be a log price--to--book ratio at time $t$ of $i$--th private company. To keep notations simple, let $\mathsf{B}_t:=(\mathsf{B}_{1,t},\dots,\mathsf{B}_{n,t})'$ be an $(n\times 1)$ book value of equity process at time $t$, $\tilde{b}_t:=(\tilde{b}_{1,t},\dots,\tilde{b}_{n,t})'$ be an $(n\times 1)$ log book value growth rate process at time $t$, $\tilde{\Delta}_t:=(\tilde{\Delta}_{1,t},\dots,\tilde{\Delta}_{n,t})'$ be an $(n\times 1)$ log dividend--to--book ratio process at time $t$, and $\tilde{m}_t:=(\tilde{m}_{1,t},\dots,\tilde{m}_{n,t})'$ be an $(n\times 1)$ log price--to--book ratio process at time $t$. Then, in vector form, equations \eqref{03007} and \eqref{03008} can be written by
\begin{equation}\label{03009}
\ln(\mathsf{B}_t)=\tilde{b}_t+\ln(\mathsf{B}_{t-1})
\end{equation}
and
\begin{equation}\label{03010}
\tilde{d}_t=\tilde{\Delta}_t+\ln(\mathsf{B}_{t-1}).
\end{equation}

Let us assume that a log spot interest rate and economic variables that affect the log book value growth rate process $\tilde{b}_t$ and the log price--to--book ratio process $\tilde{m}_t$ are placed on $(\ell\times 1)$ dimensional Vector Autoregressive Process of order $p$ (VAR$(p)$) $z_t$. The VAR$(p)$ process $z_t$ is represented by
\begin{equation}\label{03011}
z_t=C_z\psi_t+A_1z_{t-1}+\dots+A_pz_{t-p}+v_t, ~~~t=1,\dots,T,
\end{equation}
where $z_t:=(z_{1,t},\dots,z_{\ell,t})'$ is an $(\ell\times 1)$ endogenous variables process, $v_t:=(v_{1,t},\dots,v_{\ell,t})'$ is an $(\ell\times 1)$ white noise process, $C_z$ is an $(\ell\times l)$ coefficient matrix, corresponding to the exogenous variables, and $A_1,\dots,A_p$ are $(\ell\times \ell)$ coefficient matrices, corresponding to processes $z_{t-1},\dots,z_{t-p}$. Then, equation \eqref{03011} can be written by
\begin{equation}\label{03012}
z_t=C_z\psi_t+Az_{t-1}^*+v_t, ~~~t=1,\dots,T,
\end{equation}
under the real probability measure $\mathbb{P}$, where $z_{t-1}^*:=(z_{t-1}',\dots,z_{t-p}')'$ is an $(\ell p\times 1)$ economic variables process, which consists of last $p$ lagged values of the process $z_t$ and $A:=[A_1:\dots:A_p]$ is an $(\ell\times \ell p)$ stacked matrix, which consists of the coefficient matrices $A_1,\dots,A_p$.

Since a spot interest rate $r_t$ is known at time $(t-1)$, without loss of generality we suppose that log spot interest rate $\tilde{r}_t:=\ln(1+r_t)$ is placed on the first component of the process $z_{t-1}$. That is, $\tilde{r}_t=e_{\ell,1}'z_{t-1}$, where $e_{\ell,i}:=(0,\dots,0,1,0,\dots,0)'$ is an $(\ell\times 1)$ unit vector, whose $i$--th component is one and others are zero. If $\ell>1$, then other components of the process $z_t$ represent economic variables that affect the log book value growth rates and the log price--to--book ratios of the companies. 

It should be noted that by using the augmented Dickey--Fuller test one can easily confirm that for the quarterly S\&P 500 index, its log price--to--book ratio is governed by the unit--root process with drift, see data from \citeA{Nasdaq23a}. For this reason, we assume that the log price--to--book ratio process of the companies follows the unit root process with drift, that is,
\begin{equation}\label{03013}
\tilde{m}_t=C_m\psi_t+\tilde{m}_{t-1}+w_t
\end{equation}
under the real probability measure $\mathbb{P}$, where $w_t:=(w_{1,t},\dots,w_{n,t})'$ is an $(n\times 1)$ white noise process and $C_m$ is an $(n\times l)$ coefficient matrix, corresponding to the exogenous variables

If we substitute equations $\tilde{P}_t=\tilde{m}_t+\tilde{b}_t+\ln(\mathsf{B}_{t-1})$, $\tilde{P}_{t-1}=\tilde{m}_{t-1}+\ln(\mathsf{B}_{t-1})$, and $\tilde{d}_t=\tilde{\Delta}_t+\ln(\mathsf{B}_{t-1})$ into equation \eqref{03006}, then the dynamic Gordon growth model for the private companies is given by the following equation
\begin{equation}\label{03014}
\tilde{b}_t=-\tilde{m}_t+G_t\tilde{m}_{t-1}+G_tC_k\psi_t-(G_t-I_n)\tilde{\Delta}_t-h_t+G_tu_t
\end{equation}
under the real probability measure $\mathbb{P}$. Note that for $i$--th company, it follows from the above equation that
\begin{equation}\label{03015}
\tilde{b}_{i,t}=-\tilde{m}_{i,t}+g_{i,t}\tilde{m}_{i,t-1}+c_{i,k}\psi_tg_{i,t}-(g_{i,t}-1)\tilde{\Delta}_{i,t}-h_{i,t}+g_{i,t}u_{i,t}.
\end{equation}
where $c_{i,k}$ is an $i$--th row of the matrix $C_k$. Thus, if $i$--th company does not pay a dividend at time $t$, i.e., $d_{i,t}=0$, then as $\lim_{d_{i,t}\to 0}g_{i,t}=1$ and $\lim_{d_{i,t}\to 0}h_{i,t}=0$, equation \eqref{03015} becomes
\begin{equation}\label{03016}
\tilde{b}_{i,t}=-\tilde{m}_{i,t}+\tilde{m}_{i,t-1}+c_{i,k}\psi_t+u_{i,t}
\end{equation}
under the real probability measure $\mathbb{P}$. In this case, the approximate equation becomes exact, see \citeA{Battulga23b}.

As a result, combining equations \eqref{03012}, \eqref{03013}, and \eqref{03014}, one obtains a private company valuation model (system), which models the companies whose dividends are sometimes paid and sometimes not paid:
\begin{equation}\label{03017}
\begin{cases}
\tilde{b}_t=\nu_{b,t}+\Psi_{b,t}\tilde{m}_t^*+G_tu_t\\
z_t=\nu_{z,t}+Az_{t-1}^*+v_t\\
\tilde{m}_t=\nu_{m,t}+\tilde{m}_{t-1}+w_t
\end{cases}~~~\text{for}~t=1,\dots,T
\end{equation}
under the real probability measure $\mathbb{P}$, where $\nu_{b,t}:=G_tC_k\psi_t-(G_t-I_n)\tilde{\Delta}_t-h_t$ is an $(n\times 1)$ intercept process of log book value growth rate process $\tilde{b}_t$, $\Psi_{b,t}:=[-I_n:G_t]$ is an $(n\times 2n)$ matrix, whose first block matrix is $-I_n$ and second block matrix is $G_t$, $\tilde{m}_t^*:=(\tilde{m}_t',\tilde{m}_{t-1}')'$ is a $(2n\times 1)$ log price--to--book ratio process, $\nu_{z,t}:=C_z\psi_t$ is an $(\ell\times 1)$ intercept process of economic variables process $z_t$, and $\nu_{m,t}:=C_m\psi_t$ is an $(n\times 1)$ intercept process of log price--to--book ratio process $\tilde{m}_t$. In this paper, we assume that the log dividend--to--price process $\tilde{\Delta}_t$ is deterministic (exogenous variable), see equation \eqref{03165}. Let us denote a dimension of system \eqref{03017} by $\tilde{n}:=2n+\ell$. Note that system \eqref{03017} can be used for both private and public companies. However, we refer to the system as a log private company valuation model (system). For a private company valuation model, we refer to \citeA{Battulga23b}.

For log private company valuation system \eqref{03017}, because values of the log book value growth rate process $\tilde{b}_t$ and economic variables process $z_t$ are observed and values the log price--to--book ratio process $\tilde{m}_t$ is unobserved, to estimate parameters of the system, we should use the Kalman filtering. The first two lines of the system define the measurement equation and the last line of the system defines the transition equation, respectively, of the Kalman filtering, see \citeA{Lutkepohl05}.

The stochastic properties of system \eqref{03017} is governed by the random vectors $\{u_1,\dots,u_T,v_1,\dots,v_T,$ $w_1,\dots,w_T, m_0\}$. Because of the fact that if the measurement equation's white noise processes $u_t$ and $v_t$ are correlated with the transition equation's white noise process $w_t$, then the Kalman filtering will collapse, see Section 6, we assume that for $t=1,\dots,T$, the measurement equation's white noise processes $u_t$ and $v_t$, the transition equation's white noise process $w_t$, and an initial log price--to--book ratio vector $\tilde{m}_0$ are mutually independent, and follow normal distributions, namely,
\begin{equation}\label{03018}
\tilde{m}_0\sim \mathcal{N}(\mu_0,\Sigma_0), ~~~\eta\sim \mathcal{N}(0,\Sigma_{\eta\eta}),~~~w_t\sim \mathcal{N}(0,\Sigma_{ww}), ~~~\text{for}~t=1,\dots,T
\end{equation}
under the real probability measure $\mathbb{P}$, where $\eta_t:=(u_t',v_t')'$ is an $([n+\ell]\times 1)$ white noise process of the measurement equation and its covariance matrix equals 
\begin{equation}\label{03019}
\Sigma_{\eta\eta}:=\begin{bmatrix}
\Sigma_{uu} & \Sigma_{uv}\\
\Sigma_{vu} & \Sigma_{vv}
\end{bmatrix}.
\end{equation}

It is worth mentioning that one may be model the dividend process by the classic/standard method of \citeA{Merton73}, namely, the dividend process is proportional to the stock price process
\begin{equation}\label{ad001}
\tilde{d}_t=\tilde{\Delta}_t+\tilde{P}_{t-1}
\end{equation}
where $\tilde{\Delta}_t:=\ln\big(d_t/P_{t-1}\big)$ is an $(n\times 1)$ vector of dividend--to--price ratios (dividend yield) of the companies. Then, it can be shown that a book value growth rate process $\tilde{b}_t$, corresponding to dividend equation \eqref{ad001} is 
\begin{eqnarray}\label{ad005}
\tilde{b}_t&=&-\tilde{m}_t+\tilde{m}_{t-1}+G_tC_k\psi_t-(G_t-I_n)\tilde{\Delta}_t-h_t+G_tu_t\nonumber\\
&=&\nu_{b,t}+\Psi_b\tilde{m}_t^*+G_tu_t
\end{eqnarray}
under the real probability measure $\mathbb{P}$, $\Psi_b:=[-I_n:I_n]$ is an $(n\times 2n)$ matrix, whose first block matrix is $-I_n$ and second block matrix is $I_n$. Therefore, by replacing $\Psi_{b,t}$ with $\Psi_b$, all the theoretical results, which arise in the following sections are same for the both dividend models, except parameter estimations. Also, it can be shown that the dividend--to--price ratio process of the S\&P 500 index is the unit root process without drift.

\section{Option Pricing}

\citeA{Black73} developed a closed--form formula for evaluating the European call option. The formula assumes that the underlying asset follows geometric Brownian motion, but does not take dividends into account. For public companies, most stock options traded on the option exchange pay dividends at least once before they expire and for most private companies, they also pay dividends. Therefore, it is important to develop formulas for the European call and put options on dividend--paying stocks from a practical point of view. \citeA{Merton73} first time used continuous dividend in the Black--Scholes framework and obtained a similar pricing formula with the Black--Scholes formula. In this Section, we develop an option pricing model for private and public companies, where the dividend process is given by equation \eqref{03008} or \eqref{ad001}. To price the Black--Scholes call and put options, we use the risk--neutral valuation method. 

Let $T$ be a time to maturity of the Black--Scholes call and put options at time zero, and for $t=1,\dots,T$, $\xi_t:=(u_t',v_t',w_t')'$ be an $(\tilde{n}\times 1)$ white noise process of system \eqref{03017}. According to equation \eqref{03018}, $\xi_1,\dots,\xi_T$ is a random sequence of independent identically multivariate normally distributed random vectors with means of $(\tilde{n}\times 1)$ zero vector and covariance matrices of $(\tilde{n}\times \tilde{n})$ matrix $\Sigma_{\xi\xi}:=\text{diag}\{\Sigma_{\eta\eta},\Sigma_{ww}\}$. Therefore, a distribution of a white noise vector $\xi:=(\xi_1',\dots,\xi_T')'$ is given by 
\begin{equation}\label{03020}
\xi \sim \mathcal{N}\big(0,I_T\otimes\Sigma_{\xi\xi}\big),
\end{equation}
where $\otimes$ is the Kronecker product of two matrices.

Let $x:=(x_1',\dots,x_T')'$ be an $(\tilde{n}T\times 1)$ vector, which consists of all the log book value growth rates, economic variables, including the log spot interest rate, and log price--to--book ratios of the companies and whose $t$--th sub--vector is $x_t:=(\tilde{b}_t',z_t',\tilde{m}_t')'$. We define $\sigma$--fields, which play major roles in the paper: $\mathcal{F}_0:=\sigma(\mathsf{B}_0,\tilde{\Delta}_1,\dots,\tilde{\Delta}_T,z_0,\dots,z_{-p+1})$ and for $t=1,\dots,T$, $\mathcal{F}_{t}:=\mathcal{F}_0\vee\sigma(\tilde{b}_1,z_1,\dots,\tilde{b}_t,z_t)$ and $\mathcal{G}_{t}:=\mathcal{F}_t\vee \sigma(\tilde{m}_0,\dots,\tilde{m}_t)$, where for generic $\sigma$--fields $\mathcal{O}_1$ and $\mathcal{O}_2$, $\mathcal{O}_1\vee \mathcal{O}_2$ is the minimal $\sigma$--field containing them. Note that $\sigma(\mathsf{B}_0,\tilde{b}_1,\dots,\tilde{b}_t)=\sigma(\mathsf{B}_0,\mathsf{B}_1,\dots,\mathsf{B}_t)$, the $\sigma$--field $\mathcal{F}_t$ represents available information at time $t$ for the private companies, the $\sigma$--field $\mathcal{G}_t$ represents available information at time $t$ for public companies, and the $\sigma$--fields satisfy $\mathcal{F}_{t}\subset \mathcal{G}_{t}$ for $t=0,\dots,T$. For the public company, to price and hedge the options and life insurance products, we have to use the information $\mathcal{G}_t$. Therefore, one obtains pricing and hedging formulas for the public company by replacing the information $\mathcal{F}_t$, corresponding to the private company with the information $\mathcal{G}_t$. 

\subsection{Risk--Neutral Probability Measure}

To price the call and put options and equity--linked life insurance products, we need to change from the real probability measure to some risk--neutral probability measure. Let $D_t:=\exp\{-\tilde{r}_1-\dots-\tilde{r}_t\}=1\big/\prod_{s=1}^t(1+r_s)$ be a discount process, where $\tilde{r}_t=e_{\ell,1}'z_{t-1}$ is the log spot interest rate at time $t$. According to \citeA{Pliska97} (see also \citeA{Bjork20}), for all the companies, conditional expectations of a return processes $(P_{i,t}+d_{i,t})/P_{i,t-1}-1$ $(i=1,\dots,n)$ must equal the spot interest rate $r_t$ under some risk--neutral probability measure $\tilde{\mathbb{P}}$ and a filtration $\{\mathcal{G}_t\}_{t=0}^T$. Thus, it must hold
\begin{equation}\label{03021}
\tilde{\mathbb{E}}\big[(P_t+d_t)\oslash P_{t-1}\big|\mathcal{G}_{t-1}\big]=\exp\big\{\tilde{r}_ti_n\big\}
\end{equation}
for $t=1,\dots,T$, where $\tilde{\mathbb{E}}$ denotes an expectation under the risk--neutral probability measure $\mathbb{\tilde{P}}$. According to equation \eqref{03004}, condition \eqref{03021} is equivalent to the following condition
\begin{equation}\label{03022}
\tilde{\mathbb{E}}\big[\exp\big\{u_t-(\tilde{r}_ti_n-C_k\psi_t)\big\}\big|\mathcal{G}_{t-1}\big]=i_n.
\end{equation}

It should be noted that condition \eqref{03022} corresponds only to the white noise random process $u_t$. Thus, we need to impose a condition on the white noise random processes $v_t$ and $w_t$ under the risk--neutral probability measure. This condition is fulfilled by $\tilde{\mathbb{E}}[f(v_t,w_t)|\mathcal{G}_{t-1}]=\tilde{\theta}_t$ for any Borel function $f:\mathbb{R}^{n+\ell}\to\mathbb{R}^{n+\ell}$ and $\mathcal{G}_{t-1}$ measurable any $([n+\ell]\times 1)$ random vector $\tilde{\theta}_t$. Because for any admissible choices of $\tilde{\theta}_t$, condition \eqref{03022} holds, the market is incomplete. But prices of the options and equity--linked life insurance products are still consistent with the absence of arbitrage. 

In order to change from the real probability measure $\mathbb{P}$ to the risk neutral probability measure $\mathbb{\tilde{P}}$, we define the following state price density process:
\begin{equation}\label{ad007}
M_t=K_0\times L_t,~~~t=1,\dots,T,
\end{equation}
where
\begin{equation}\label{ad008}
K_0:=\exp\bigg\{\theta_0'\Sigma_0^{-1}\tilde{m}_0+\frac{1}{2}\theta_0'\Sigma_0^{-1}\theta_0\bigg\}
\end{equation}
and
\begin{equation}\label{ad009}
L_t:=\exp\bigg\{\sum_{s=1}^t\theta_s'\mathsf{G}_s^{-1}\Sigma_{\xi\xi}^{-1}\xi_s+\frac{1}{2}\sum_{s=1}^t\theta_s'\mathsf{G}_s^{-1}\Sigma_{\xi\xi}^{-1}\mathsf{G}_s^{-1}\theta_s\bigg\}
\end{equation}
with $L_0=1$, where for $s=1,\dots,T$, the random vector $\theta_s$ is measurable with respect to $\sigma$--field $\mathcal{G}_{s-1}$ and $\mathsf{G}_s:=\text{diag}\{G_s,I_\ell,I_n\}$ is an $(\tilde{n}\times \tilde{n})$ block diagonal matrix. Since $\mathbb{E}\big[M_T\big|\mathcal{F}_0\big]=\mathbb{E}\big[\mathbb{E}[M_T|\mathcal{G}_0]\big|\mathcal{F}_0\big]=1$ and for all $\omega\in\Omega$, $M_T(\omega)>0$,
\begin{equation}\label{ad010}
\tilde{\mathbb{P}}\big[A\big|\mathcal{F}_0\big]=\int_AM_T(\omega|\mathcal{F}_t)d\mathbb{P}(\omega|\mathcal{F}_0), ~~~\text{for all}~A\in\mathcal{G}_T
\end{equation}
becomes probability measure. Then, it can be shown that a relative entropy and variance of the state price density process at time $T$ are given by
\begin{eqnarray}\label{ad011}
I(\tilde{\mathbb{P}},\mathbb{P}\big|\mathcal{F}_0)&:=&\tilde{\mathbb{E}}\big[\ln(M_T)\big|\mathcal{F}_0\big]=\frac{1}{2}\tilde{\mathbb{E}}\bigg[\theta_0'\Sigma_0^{-1}\theta_0+\sum_{s=1}^T\theta_s'\mathsf{G}_s^{-1}\Sigma_{\xi\xi}^{-1}\mathsf{G}_s^{-1}\theta_s\bigg|\mathcal{F}_0\bigg]
\end{eqnarray}
and
\begin{equation}\label{ad012}
\text{Var}\big[M_T\big|\mathcal{F}_0\big]=\tilde{\mathbb{E}}\bigg[\exp\bigg\{\theta_0'\Sigma_0^{-1}\theta_0+\sum_{s=1}^T\theta_s'\mathsf{G}_s^{-1}\Sigma_{\xi\xi}^{-1}\mathsf{G}_s^{-1}\theta_s\bigg\}\bigg|\mathcal{F}_0\bigg]-1,
\end{equation}
see \citeA{Battulga24a}. As a result, according to \citeA{Battulga24a}, for each $t=0,\dots,T$, the optimal Girsanov kernel process $\theta_t^*$, which minimizes the variance of the state price density process and relative entropy is obtained by
\begin{equation}\label{03023}
\theta_t^*=\Theta_t \bigg(\tilde{r}_ti_n-C_k\psi_t-\frac{1}{2}\mathcal{D}[\Sigma_{uu}]\bigg),~~~t=1,\dots,T~~~\text{and}~~~\theta_0^*=0,
\end{equation}
where $\Theta_t=\big[G_t:(\mathcal{S}\Sigma_{uu}^{-1})'\big]'$ is an $(\tilde{n}\times n)$ matrix, $\mathcal{S}:=[\Sigma_{vu}':0]'$ is a $([n+\ell]\times n)$ matrix, and for a generic square matrix $O$, $\mathcal{D}[O]$ denotes a vector, consisting of diagonal elements of the matrix $O$. 

Consequently, the log price process $\tilde{P}_t$ is given by
\begin{equation}\label{03024}
\tilde{P}_t=G_t\bigg(\tilde{P}_{t-1}-\tilde{d}_t+\tilde{r}_t i_n-\frac{1}{2}\mathcal{D}[\Sigma_{uu}]\bigg)+\tilde{d}_t-h_t+G_t\tilde{u}_t
\end{equation}
and system \eqref{03017} becomes
\begin{equation}\label{03025}
\begin{cases}
\tilde{b}_t=\tilde{\nu}_{b,t}+\Psi_{b,t}\tilde{m}_t^*+E_tz_{t-1}^*+G_t\tilde{u}_t,\\
z_t=\tilde{\nu}_{z,t}+\tilde{A}z_{t-1}^*+\tilde{v}_t,\\
\tilde{m}_t=\nu_{m,t}+\tilde{m}_{t-1}+\tilde{w}_t,
\end{cases}~~~\text{for}~t=1,\dots,T
\end{equation}
under the risk--neutral probability measure $\tilde{\mathbb{P}}$, where $\tilde{\nu}_{b,t}:=-(G_t-I_n)\tilde{\Delta}_t-\frac{1}{2}G_t\mathcal{D}[\Sigma_{uu}]-h_t$ is an $(n\times 1)$ intercept process of the log book value growth rate process $\tilde{b}_t$, $E_t:=[G_ti_ne_{\ell,1}':0:\dots:0]$ is an ($n\times \ell p$) coefficient matrix of the economic variables process $z_{t-1}^*$, $\tilde{\nu}_{z,t}:=C_z\psi_t-\Sigma_{vu}\Sigma_{uu}^{-1}\big(C_k\psi_t+\frac{1}{2}\mathcal{D}[\Sigma_{uu}]\big)$ is an $(\ell\times 1)$ intercept process of the economic variables process $z_t$, and $\tilde{A}:=[\tilde{A}_1:A_2:\dots:A_p]$ with $\tilde{A}_1:=A_1+\Sigma_{vu}\Sigma_{uu}^{-1}i_ne_{\ell,1}'$ is an ($\ell\times \ell p$) coefficient matrix, see \citeA{Battulga24a}.

It is worth mentioning that because the white noise processes $(u_t',v_t')'$ and $w_t$ and initial price--to--book ratio $\tilde{m}_0$ are independent under the real probability measure $\mathbb{P}$, the white noise processes $(\tilde{u}_t',\tilde{v}_t')'$ and $\tilde{w}_t$ and initial price--to--book ratio $\tilde{m}_0$ are also independent under the risk--neutral probability measure $\mathbb{\tilde{P}}$ and the distributions of the log price--to--book ratio $\tilde{m}_t$ are same for both probability measures, see \citeA{Battulga24a}. Also, note that a joint distribution of a random vector $\tilde{\xi}:=(\tilde{\xi}_1',\dots,\tilde{\xi}_T')'$ with $\tilde{\xi}_t:=(\tilde{u}_t',\tilde{v}_t',\tilde{w}_t')'$ equals the joint distribution of the random vector $\xi=(\xi_1',\dots,\xi_T')'$, that is,
\begin{equation}\label{03026}
\tilde{\xi}\sim \mathcal{N}\big(0,I_T\otimes \Sigma_{\xi\xi}\big)
\end{equation}
under the risk--neutral probability measure $\mathbb{\tilde{P}}$.

By comparing systems \eqref{03017} and \eqref{03025}, one can deduce that for the first and second lines of system (20), the term $C_k\psi_t$, which is a part of the log required rate of return process vanishes and the terms, which are related to the log spot interest rate and covariance matrix of the random vector $u_t$ emerge. It is an analog of the well--known fact that to price the European options, if a stock price is modeled by geometric Brownian motion, the drift process of the geometric Brownian motion is replaced by a risk--free interest rate process under a risk--neutral probability measure. However, the last line of system \eqref{03025} does not change because the white noise processes $\eta_t=(u_t',v_t')'$ and $w_t$ and initial price--to--book ratio $\tilde{m}_0$ are independent under the real probability measure $\mathbb{P}$. Also, we can conclude that because of the change of probability measure, the log price at time $t$ is changed, see equations \eqref{03006} and \eqref{03024}.

To obtain a joint distribution of the random process $x_t$, let us write system \eqref{03025} in the following form
\begin{equation}\label{03027}
\begin{cases}
\tilde{b}_t=\tilde{\nu}_{b,t}+\Psi_{b,t}\tilde{m}_t^*+E_tz_{t-1}^*+G_t\tilde{u}_t\\
z_t^*=\tilde{\nu}_{z,t}^*+\tilde{A}^*z_{t-1}^*+\tilde{v}_t^*\\
\tilde{m}_t^*=\nu_{m,t}^*+C\tilde{m}_{t-1}^*+\tilde{w}_t^*
\end{cases}~~~\text{for}~t=1,\dots,T
\end{equation}
under the risk--neutral probability measure $\mathbb{\tilde{P}}$, where $\tilde{\nu}_{z,t}^*:=(\tilde{\nu}_{z,t}',0,\dots,0)'$ is an $(\ell p\times 1)$ intercept process and $\tilde{v}_t^*:=(\tilde{v}_t',0,\dots,0)'$ is an $(\ell p\times 1)$ white noise process of the economic variables process $z_t^*$, and $\nu_{m,t}^*:=\big(\nu_{m,t}',0\big)'$ is a $(2n\times 1)$ intercept process and $\tilde{w}_t^*:=(\tilde{w}_t',0)'$ is a $(2n \times 1)$ white noise process of the log price--to--book ratio process $\tilde{m}_t^*$, and
\begin{equation}\label{03028}
\tilde{A}^*:=\begin{bmatrix}
\tilde{A}_1 & \dots & A_{p-1} & A_p\\
I_\ell & \dots & 0 & 0\\
\vdots & \ddots & \vdots & \vdots\\
0 & \dots & I_\ell & 0
\end{bmatrix},~~~\text{and}~~~
C:=\begin{bmatrix}
I_n & 0\\
I_n & 0
\end{bmatrix}
\end{equation}
are $(\ell p\times \ell p)$ and $(2n\times 2n)$ matrices, respectively. We denote a dimension of system \eqref{03027} by $\tilde{n}^*:=3n+\ell p$. Here we assume that initial value of the log price--to--book ratio process equals $\tilde{m}_0^*:=(\tilde{m}_0',\tilde{m}_0')'$. Consequently, its expectation and covariance matrix are $\tilde{m}_{0|0}^*:=\mathbb{E}[\tilde{m}_0^*|\mathcal{F}_0]=(\mu_0',\mu_0')'$ and $\Sigma(\tilde{m}_0^*|0):=\text{Var}[\tilde{m}_0^*|\mathcal{F}_0]=\mathsf{E}_2\otimes\Sigma$, respectively, where $\mathsf{E}_n$ is an ($n\times n$) matrix, which consists of ones. One can write system \eqref{03027} in VAR(1) form:
\begin{equation}\label{03029}
Q_{0,t}x_t^*=\tilde{\nu}_t^*+Q_{1,t}x_{t-1}^*+\mathsf{G}_t^*\tilde{\xi}_t^*
\end{equation} 
under the risk--neutral probability measure $\mathbb{\tilde{P}}$, where $x_t^*:=\big(\tilde{b}_t',(z_t^*)',(\tilde{m}_t^*)'\big)'$, $\tilde{\nu}_t^*:=\big(\tilde{\nu}_{b,t}',(\tilde{\nu}_{z,t}^*)',(\tilde{\nu}_{m,t}^*)'\big)'$, and $\tilde{\xi}_t^*:=\big(\tilde{u}_t',(\tilde{v}_t^*)',(\tilde{w}_t^*)'\big)'$ are $(\tilde{n}^*\times 1)$ VAR(1) process, intercept process, and white noise process, respectively,
\begin{equation}\label{03030}
Q_{0,t}:=\begin{bmatrix}
I_n & 0 & -\Psi_{b,t}\\
0 & I_{\ell p} & 0\\
0 & 0 & I_{2n}
\end{bmatrix},~~~
Q_{1,t}:=\begin{bmatrix}
0 & E_t & 0\\
0 & \tilde{A}^* & 0\\
0 & 0 & C
\end{bmatrix}~~~\text{and}~~~
\mathsf{G}_t^*=\begin{bmatrix}
G_t & 0 & 0\\
0 & \mathsf{J}_z'\mathsf{J}_z & 0\\
0 & 0 & \mathsf{J}_m'\mathsf{J}_m
\end{bmatrix}
\end{equation}
are $(\tilde{n}^*\times \tilde{n}^*)$ coefficient matrices, $\mathsf{J}_z:=[I_{\ell}:0]$ is an $(\ell\times \ell p)$ matrix, whose first block matrix is $I_\ell$ and other blocks are zero, and $\mathsf{J}_m:=[I_{n}:0]$ is an $(n \times 2n)$ matrix, whose first block matrix is $I_n$ and second block is zero. By repeating equation \eqref{03029}, one gets that for $s=t+1,\dots,T$,
\begin{equation}\label{03031}
x_s^*=\Pi_{t,s}^*x_t^*+\sum_{\beta=t+1}^s\Pi_{\beta,s}^*\tilde{\nu}_\beta^*+\sum_{\beta=t+1}^s\Pi_{\beta,s}^*\mathsf{G}_\beta^*\tilde{\xi}_\beta^*,
\end{equation}
where the coefficient matrices are
\begin{equation}\label{03032}
\Pi_{\beta,s}^*:=\prod_{\alpha=\beta+1}^sQ_{0,\alpha}^{-1}Q_{1,\alpha}=\begin{bmatrix}
0 & E_s(\tilde{A}^*)^{s-\beta-1} & \Psi_{b,s}C^{s-\beta}\\
0 & (\tilde{A}^*)^{s-\beta} & 0\\
0 & 0 & C^{s-\beta}
\end{bmatrix}
\end{equation}
for $\beta=0,\dots,s-1$ and 
\begin{equation}\label{03033}
\Pi_{s,s}^*:=Q_{0,s}^{-1}=\begin{bmatrix}
I_n & 0 & \Psi_{b,s}\\
0 & I_{\ell p} & 0\\
0 & 0 & I_{2n}
\end{bmatrix}.
\end{equation}
By using the matrices $\mathsf{J}_z$ and $\mathsf{J}_m$, one can transform the intercept processes $\tilde{\nu}_{z,t}$ and $\nu_{m,t}$ into the intercept processes $\tilde{\nu}_{z,t}^*$ and $\nu_{m,t}^*$, i.e., $\tilde{\nu}_{z,t}^*=\mathsf{J}_z'\tilde{\nu}_{z,t}$ and $\nu_{m,t}^*=\mathsf{J}_m'\nu_{m,t}$. Similarly, it holds $\tilde{v}_t^*=\mathsf{J}_z'\tilde{v}_t$ and $\tilde{w}_t^*=\mathsf{J}_m'\tilde{w}_t$. Now, let $\mathsf{J}_x=\text{diag}\{I_n,\mathsf{J}_z,\mathsf{J}_m\}$ be an $(\tilde{n}\times \tilde{n}^*)$ block diagonal matrix. Then, since $\tilde{\nu}_\beta^*=\mathsf{J}_x'\tilde{\nu}_\beta$ with $\tilde{\nu}_\beta:=\big(\tilde{\nu}_{b,\beta}',\tilde{\nu}_{z,\beta}',\nu_{m,\beta}'\big)'$ and $\tilde{\xi}_\beta^*=\mathsf{J}_x'\tilde{\xi}_\beta$, equation \eqref{03031} can be written by
\begin{equation}\label{03034}
x_s^*=\Pi_{t,s}^*x_t^*+\sum_{\beta=t+1}^s\Pi_{\beta,s}^*\mathsf{J}_x'\tilde{\nu}_\beta+\sum_{\beta=t+1}^s\Pi_{\beta,s}^*\mathsf{G}_\beta^*\mathsf{J}_x'\tilde{\xi}_\beta
\end{equation}
under the risk--neutral probability measure $\mathbb{\tilde{P}}$. The matrix $\mathsf{J}_x$ can be used to extract the process $x_s$ from the VAR(1) process $x_s^*$, i.e., $x_s=\mathsf{J}_xx_s^*$. Consequently, since $\mathsf{J}_x'\mathsf{J}_x\mathsf{G}_\beta^*=\mathsf{G}_\beta^*$, we have
\begin{equation}\label{03035}
x_s=\mathsf{J}_x\Pi_{t,s}^*x_t^*+\sum_{\beta=t+1}^s\Pi_{\beta,s}\tilde{\nu}_\beta+\sum_{\beta=t+1}^s\Pi_{\beta,s}\mathsf{G}_\beta\tilde{\xi}_\beta,~~~s=t+1,\dots,T
\end{equation}
under the risk--neutral probability measure $\mathbb{\tilde{P}}$, where $\Pi_{\beta,s}:=\mathsf{J}_x\Pi_{\beta,s}^*\mathsf{J}_x'$ is an $(\tilde{n}\times \tilde{n})$ matrix and $\mathsf{G}_\beta=\mathsf{J}_x\mathsf{G}_\beta^*\mathsf{J}_x'=\text{diag}\{G_\beta,I_\ell,I_n\}$ is the $(\tilde{n}\times \tilde{n})$ block diagonal matrix. Therefore, conditional on the information $\mathcal{G}_t$, for $t+1\leq s,s_1,s_2\leq T$, a conditional expectation at time $s$ and conditional covariance matrix at time $s_1$ and $s_2$ of the process $x_t$ are given by the following equations
\begin{equation}\label{03036}
\tilde{\mu}_{s|t}(\mathcal{G}_t):=\mathbb{\tilde{E}}\big[x_s\big|\mathcal{G}_t\big]=\mathsf{J}_x\Pi_{t,s}^* x_t^*+\sum_{\beta=t+1}^s\Pi_{\beta,s}\tilde{\nu}_\beta
\end{equation}
and
\begin{equation}\label{03037}
\Sigma_{s_1,s_2|t}(\mathcal{G}_t):=\widetilde{\text{Cov}}\big[x_{s_1},x_{s_2}\big|\mathcal{G}_t\big]=\sum_{\beta=t+1}^{s_1\wedge s_2}\Pi_{\beta,s_1}\mathsf{G}_\beta\Sigma_{\xi\xi}\mathsf{G}_\beta\Pi_{\beta,s_2}',
\end{equation}
where $s_1\wedge s_2$ is a minimum of $s_1$ and $s_2$. 

\subsection{Conditional Distribution for Given $\mathcal{F}_t$}

To price the options and the equity--linked life insurance products, we need a distribution of a random vector $\bar{x}_t^c:=(x_{t+1}',\dots,x_T')'$ given $\mathcal{F}_t$. For this reason, let $y_t:=(\tilde{b}_t',z_t')'$ be an $([n+\ell]\times 1)$ random process of the observed variables, $\mathsf{J}_y:=[I_{n+\ell}:0]$ be an $([n+\ell]\times\tilde{n}^*)$ matrix, which is used to extract the random process $y_t$ from the random process $x_t^*$, $\mathsf{J}_{m^*}:=[0:I_{2n}]$ be a $(2n\times\tilde{n}^*)$ matrix, which is used to extract the random process $\tilde{m}_t^*$ from the random process $x_t^*$, $y_t^*:=(\tilde{b}_t',(z_t^*)')'$ be an $([n+\ell p]\times 1)$ random process of the observed variables, and $\Pi_{t,s}^*=\big[\Pi_{t,s}^{*y}:\Pi_{t,s}^{*m}\big]$ be a partition of the matrix $\Pi_{t,s}^*$, corresponding to the processes $y_t^*$ and $\tilde{m}_t^*$. Then, equation \eqref{03034} is written by
\begin{equation}\label{03038}
x_t^*=\Pi_{0,t}^{*y}y_0^*+\Pi_{0,t}^{*m}\tilde{m}_0^*+\sum_{s=1}^t\Pi_{s,t}^*\mathsf{J}_x'\tilde{\nu}_s+\sum_{s=1}^t\Pi_{s,t}^*\mathsf{G}_s^*\mathsf{J}_x'\tilde{\xi}_s
\end{equation}
Because of the fact that $y_t=\mathsf{J}_yx_t^*$, $\tilde{m}_t^*=\mathsf{J}_{m^*}x_t^*$, the random vectors $\tilde{m}_0^*$ and $(\tilde{\xi}_1',\dots,\tilde{\xi}_T')'$ are independent given $\mathcal{F}_0$, $\widetilde{\text{Var}}\big[\tilde{m}_0^*\big|\mathcal{F}_0\big]=\Sigma(\tilde{m}_0^*|0)$, and
\begin{equation}\label{03039}
\Sigma_{s_1,s_2|t}^*(\mathcal{G}_t):=\widetilde{\text{Cov}}\big[x_{s_1}^*,x_{s_2}^*\big|\mathcal{G}_t\big]=\sum_{\beta=t+1}^{s_1\wedge s_2}\Pi_{\beta,s_1}^*\mathsf{G}_\beta^*\mathsf{J}_x'\Sigma_{\xi\xi}\mathsf{J}_x\mathsf{G}_\beta^*(\Pi_{\beta,s_2}^*)',
\end{equation}
covariances between vectors of observed variables at times $1,\dots,t$ and log price--to--book ratio at time $t$ are obtained by
\begin{equation}\label{03040}
\widetilde{\text{Cov}}\big[y_{s_1},y_{s_2}\big|\mathcal{F}_0\big]=\mathsf{J}_y\Big(\Pi_{0,s_1}^{*m}\Sigma(\tilde{m}_0^*|0)(\Pi_{0,s_2}^{*m})'+\Sigma_{s_1,s_2|0}^*(\mathcal{G}_t)\Big)\mathsf{J}_y',~~~s_1,s_2=1,\dots,t,
\end{equation}
\begin{equation}\label{03041}
\widetilde{\text{Cov}}\big[\tilde{m}_t^*,y_{s}\big|\mathcal{F}_0\big]=\mathsf{J}_{m^*}\Big(\Pi_{0,t}^{*m}\Sigma(\tilde{m}_0^*|0)(\Pi_{0,s}^{*m})'+\Sigma_{t,s|0}^*(\mathcal{G}_t)\Big)\mathsf{J}_y', ~~~s=1,\dots,t
\end{equation}
and
\begin{equation}\label{03042}
\widetilde{\text{Var}}\big[\tilde{m}_t^*\big|\mathcal{F}_0\big]=\mathsf{J}_{m^*}\Big(\Pi_{0,t}^{*m}\Sigma(\tilde{m}_0^*|0)(\Pi_{0,t}^{*m})'+\Sigma_{t,t|0}^*(\mathcal{G}_t)\Big)\mathsf{J}_{m^*}'.
\end{equation}
Also, since $\tilde{\mathbb{E}}\big[\tilde{m}_0^*\big|\mathcal{F}_0\big]=\tilde{m}_{0|0}^*$, conditional on $\mathcal{F}_0$, expectations of the processes $y_t$ and $\tilde{m}_t^*$ are given by
\begin{equation}\label{03043}
\tilde{\mathbb{E}}\big[y_t\big|\mathcal{F}_0\big]=\mathsf{J}_y\bigg(\Pi_{0,t}^{*y}y_0^*+\Pi_{0,t}^{*m}\tilde{m}_{0|0}^*+\sum_{s=1}^t\Pi_{s,t}^*\mathsf{J}_x'\tilde{\nu}_s\bigg)
\end{equation}
and
\begin{equation}\label{03044}
\tilde{\mathbb{E}}\big[\tilde{m}_t^*\big|\mathcal{F}_0\big]=\mathsf{J}_{m^*}\bigg(\Pi_{0,t}^{*y}y_0^*+\Pi_{0,t}^{*m}\tilde{m}_{0|0}^*+\sum_{s=1}^t\Pi_{s,t}^*\mathsf{J}_x'\tilde{\nu}_s\bigg),
\end{equation}
respectively. To obtain conditional distribution of the log price--to--book ratio at time $t$, $\tilde{m}_t$, given the information $\mathcal{F}_t$, let us introduce the following notations: $\bar{y}_t:=(y_1',\dots,y_t')'$ be an $([(n+\ell)t]\times 1)$ vector, which consists of the observed variables, $\widetilde{\text{Cov}}[y_t|\mathcal{F}_0]$ is an $([(n+\ell)t]\times [(n+\ell)t])$ covariance matrix of the random vector of observed variables $\bar{y}_t$, whose elements are calculated by equation \eqref{03040}, $\widetilde{\text{Cov}}[\tilde{m}_t^*,\bar{y}_t|\mathcal{F}_0]$ is a $(2n\times [(n+\ell)t])$ covariance matrix between the log price--to--book ratio process at time $t$, $\tilde{m}_t^*$ and the random vector of observed variables $\bar{y}_t$, whose elements are calculated by equation \eqref{03041}, and $\mathbb{\tilde{E}}\big[\bar{y}_t\big|\mathcal{F}_0\big]$ is an $([(n+\ell)t]\times 1)$ expectation of the random vector of the observed variables $\bar{y}_t$, whose elements are calculated by equation \eqref{03043}. As a result, according to the well--known conditional distribution formula of a multivariate normal random vector, one gets that conditional on the information $\mathcal{F}_t$, a distribution of the log price--to--book ratio process at time $t$ is given by
\begin{equation}\label{03045}
\tilde{m}_t^*~|~\mathcal{F}_t\sim \mathcal{N}\big(\tilde{\tilde{m}}_{t|t}^*,\Sigma(\tilde{m}_t^*|t)\big),
\end{equation}
under the risk--neutral probability measure $\mathbb{\tilde{P}}$, where
\begin{equation}\label{03046}
\tilde{\tilde{m}}_{t|t}^*:=\tilde{\mathbb{E}}\big[\tilde{m}_t^*\big|\mathcal{F}_0\big]+\widetilde{\text{Cov}}[\tilde{m}_t^*,\bar{y}_t|\mathcal{F}_0]\widetilde{\text{Cov}}[\bar{y}_t|\mathcal{F}_0]^{-1}\big(\bar{y}_t-\mathbb{\tilde{E}}[\bar{y}_t|\mathcal{F}_0]\big)
\end{equation}
and
\begin{equation}\label{03047}
\Sigma(\tilde{m}_t^*|t):=\widetilde{\text{Var}}[\tilde{m}_t^*|\mathcal{F}_0]-\widetilde{\text{Cov}}[\tilde{m}_t^*,\bar{y}_t|\mathcal{F}_0]\widetilde{\text{Cov}}[\bar{y}_t|\mathcal{F}_0]^{-1}\widetilde{\text{Cov}}[\tilde{m}_t^*,\bar{y}_t|\mathcal{F}_0]'
\end{equation}
are conditional expectation and covariance matrix of the log price--to--book ratio process $\tilde{m}_t$ given the information $\mathcal{F}_t$. By equation \eqref{03035}, we have that
\begin{equation}\label{03048}
x_s=\mathsf{J}_x\Pi_{t,s}^{*y}y_t^*+\mathsf{J}_x\Pi_{t,s}^{*m}\tilde{m}_t^*+\sum_{\beta=t+1}^s\Pi_{\beta,s}\tilde{\nu}_\beta+\sum_{\beta=t+1}^s\Pi_{\beta,s}\mathsf{G}_\beta\tilde{\xi}_\beta,~~~s=t+1,\dots,T
\end{equation}
Hence, conditional on the information $\mathcal{F}_t$, for $t+1\leq s,s_1,s_2\leq T$, a conditional expectation at time $s$ and conditional covariance matrix at time $s_1$ and $s_2$ of the process $x_t$ are given by the following equations
\begin{equation}\label{03049}
\tilde{\mu}_{s|t}(\mathcal{F}_t):=\mathbb{\tilde{E}}\big[x_s\big|\mathcal{F}_t\big]=\mathsf{J}_x\Pi_{t,s}^{*y} y_t^*+\mathsf{J}_x\Pi_{t,s}^{*m} \tilde{\tilde{m}}_{t|t}^*+\sum_{\beta=t+1}^s\Pi_{\beta,s}\tilde{\nu}_\beta
\end{equation}
and 
\begin{equation}\label{03050}
\Sigma_{s_1,s_2|t}(\mathcal{F}_t):=\widetilde{\text{Cov}}\big[x_{s_1},x_{s_2}\big|\mathcal{F}_t\big]=\mathsf{J}_x\Pi_{t,s_1}^{*m} \Sigma(\tilde{m}_t^*|t)(\Pi_{t,s_2}^{*m})'\mathsf{J}_x'+\sum_{\beta=t+1}^{s_1\wedge s_2}\Pi_{\beta,s_1}\mathsf{G}_\beta\Sigma_{\xi\xi}\mathsf{G}_\beta\Pi_{\beta,s_2}'.
\end{equation}
Note that because the log price--to--book ratio process $\tilde{m}_t$ is unobserved for the information $\mathcal{F}_t$, differences arise in expectations \eqref{03036} and \eqref{03049} and covariance matrices \eqref{03037} and \eqref{03050}. Consequently, due to equation \eqref{03048}, conditional on the information $\mathcal{F}_t$, a joint distribution of the random vector $\bar{x}_t^c$ is
\begin{equation}\label{03051}
\bar{x}_t^c~|~\mathcal{F}_t\sim \mathcal{N}\big(\tilde{\mu}_t^c(\mathcal{F}_t),\Sigma_t^c(\mathcal{F}_t)\big),~~~t=0,\dots,T-1
\end{equation}
under the risk--neutral probability measure $\tilde{\mathbb{P}}$, where $\tilde{\mu}_t^c(\mathcal{F}_t):=\big(\tilde{\mu}_{t+1|t}(\mathcal{F}_t),\dots,\tilde{\mu}_{T|t}(\mathcal{F}_t)\big)'$ is a conditional expectation and $\Sigma_t^c(\mathcal{F}_t):=\big(\Sigma_{s_1,s_2|t}(\mathcal{F}_t)\big)_{s_1,s_2=t+1}^T$ is a conditional covariance matrix of the random vector $\bar{x}_t^c$ and are calculated by equations \eqref{03049} and \eqref{03050}, respectively. 

\subsection{Forward Probability Measure}

According to \citeA{Geman95}, cleaver change of probability measure leads to a significant reduction in the computational burden of derivative pricing. A frequently used probability measure that reduces the computational burden is the forward probability measure and to price and hedge the European call and put options and equity--linked life insurance products, we will apply it. To define the forward probability measure, we need to zero--coupon bond. It is the well--known fact that conditional on $\mathcal{F}_t$, price at time $t$ of zero--coupon bond paying face value 1 at time $u$ is $B_{t,u}(\mathcal{F}_t):=\frac{1}{D_t}\mathbb{\tilde{E}}\big[D_u\big|\mathcal{F}_t\big]$. For $u=t+1,\dots,T$, the $(t,u)$--forward probability measure is defined by
\begin{equation}\label{03052}
\mathbb{\hat{P}}_{t,u}\big[A\big|\mathcal{F}_t\big]:=\frac{1}{D_tB_{t,u}(\mathcal{F}_t)}\int_AD_u\mathbb{\tilde{P}}\big[\omega|\mathcal{F}_t\big]~~~\text{for all}~A\in \mathcal{G}_T.
\end{equation}
Recall that the log spot interest rate is given by $\tilde{r}_s=e_{\ell,1}'z_{s-1}$. Therefore, a negative exponent of $D_u/D_t$ in the zero--coupon bond formula is represented by 
\begin{equation}\label{03053}
\sum_{\beta=t+1}^u\tilde{r}_\beta=\tilde{r}_{t+1}+e_{\ell,1}'J_z\Bigg[\sum_{\beta=t+1}^{u-1}J_{\beta|t}\Bigg]\bar{x}_t^c=\tilde{r}_{t+1}+\gamma_{t,u}'\bar{x}_t^c
\end{equation}
with convention $\sum_{i=a+1}^ab_i=0$ for $a\in \mathbb{N}\cup{\{0\}}$, where $J_z:=[0:I_\ell:0]$ is an $(\ell\times \tilde{n})$ matrix, whose second block matrix equals $I_\ell$ and other two blocks are zero and it can be used to extract the random process $z_s$ from the random process $x_s$, $J_{\beta|t}:=[0:I_{\tilde{n}}:0]$ is an $(\tilde{n}\times [\tilde{n}(T-t)])$ matrix, whose $(\beta-t)$--th block matrix equals $I_{\tilde{n}}$ and others are zero and it can be used to extract the random vector $x_\beta$ from the random vector $\bar{x}_t^c$, and $\gamma_{t,u}':=e_{\ell,1}'J_z\sum_{\beta=t+1}^{u-1}J_{\beta|t}$ is a $(1\times [\tilde{n}(T-t)])$ vector. Therefore, due to equation \eqref{03051}, two times of negative exponent of the price at time $t$ of the zero--coupon bond is represented by 
\begin{eqnarray}\label{03054}
&&2\sum_{s=t+1}^u\tilde{r}_s+\big(\bar{x}_t^c-\tilde{\mu}_t^c(\mathcal{F}_t)\big)'\big(\Sigma_t^c(\mathcal{F}_t)\big)^{-1}\big(\bar{x}_t^c-\tilde{\mu}_t^c(\mathcal{F}_t)\big)\nonumber\\
&&=\Big(\bar{x}_t^c-\tilde{\mu}_t^c(\mathcal{F}_t)+\Sigma_t^c(\mathcal{F}_t)\gamma_{t,u}\Big)'\big(\Sigma_t^c(\mathcal{F}_t)\big)^{-1}\Big(\bar{x}_t^c-\tilde{\mu}_t^c(\mathcal{F}_t)+\Sigma_t^c(\mathcal{F}_t)\gamma_{t,u}\Big)\\
&&+2\big(\tilde{r}_{t+1}+\gamma_{t,u}'\tilde{\mu}_t^c(\mathcal{F}_t)\big)-\gamma_{t,u}'\Sigma_t^c(\mathcal{F}_t)\gamma_{t,u}.\nonumber
\end{eqnarray}
As a result, for given $\mathcal{F}_t$, the price at time $t$ of the zero--coupon bond is
\begin{equation}\label{03055}
B_{t,u}(\mathcal{F}_t)=\exp\bigg\{-\tilde{r}_{t+1}-\gamma_{t,u}'\tilde{\mu}_t^c(\mathcal{F}_t)+\frac{1}{2}\gamma_{t,u}'\Sigma_t^c(\mathcal{F}_t)\gamma_{t,u}\bigg\}.
\end{equation}
Consequently, conditional on the information $\mathcal{F}_t$, a joint distribution of the random vector $\bar{x}_t^c$ is given by
\begin{equation}\label{03056}
\bar{x}_t^c~|~\mathcal{F}_t\sim \mathcal{N}\big(\hat{\mu}_{t,u}^c(\mathcal{F}_t),\Sigma_t^c(\mathcal{F}_t)\big),~~~t=0,\dots,T-1
\end{equation}
under the $(t,u)$--forward probability measure $\hat{\mathbb{P}}_{t,u}$, where $\hat{\mu}_{t,u}^c(\mathcal{F}_t):=\tilde{\mu}_t^c(\mathcal{F}_t)-\Sigma_t^c(\mathcal{F}_t)\gamma_{t,u}$ and $\Sigma_t^c(\mathcal{F}_t)$ are conditional expectation and conditional covariance matrix, respectively, of the random vector $\bar{x}_t^c$. As we compare equations \eqref{03051} and \eqref{03056}, we conclude that because of the forward probability measure, the expectation vector, corresponding to the risk--neutral probability measure of the random vector $\bar{x}_t^c$ is changed by the additional term, while covariance matrices of the random vector $\bar{x}_t^c$ are same for both probability measures. Therefore, for $s=t+1,\dots,T,$ $(s-t)$--th block vector of the conditional expectation $\hat{\mu}_t^c(\mathcal{F}_t)$ is given by
\begin{equation}\label{03058}
\hat{\mu}_{s|t,u}(\mathcal{F}_t):=J_{s|t}\hat{\mu}_{t,u}^c(\mathcal{F}_t)=\tilde{\mu}_{s|t}(\mathcal{F}_t)-\sum_{\beta=t+1}^{u-1}\big(\Sigma_{s,\beta|t}(\mathcal{F}_t)\big)_{n+1},
\end{equation}
where for a generic matrix $O$, we denote its $j$--th column by $(O)_j$. Also, it is clear that
\begin{equation}\label{03059}
\gamma_{t,u}'\tilde{\mu}_t^c(\mathcal{F}_t)=\sum_{\beta=t+1}^{u-1}\big(\tilde{\mu}_{\beta|t}(\mathcal{F}_t)\big)_{n+1}=\sum_{\beta=t+1}^{u-1}\big(\tilde{\mu}_{\beta|t}^z(\mathcal{F}_t)\big)_1
\end{equation}
and
\begin{equation}\label{03060}
\gamma_{t,u}'\Sigma_t^c(\mathcal{F}_t)\gamma_{t,u}=\sum_{\alpha=t+1}^{u-1}\sum_{\beta=t+1}^{u-1}\big(\Sigma_{\alpha,\beta|t}(\mathcal{F}_t)\big)_{n+1,n+1},
\end{equation}
where $\tilde{\mu}_{\beta|t}^z(\mathcal{F}_t):=J_z\tilde{\mu}_{\beta|t}(\mathcal{F}_t)$ is an expectation of the random vector $z_\beta$ under the risk--neutral probability measure $\mathbb{\tilde{P}}$, for a generic vector $o$, we denote its $j$--th element by $(o)_j$, and for a generic square matrix $O$, we denote its $(i,j)$--th element by $(O)_{i,j}$. If we substitute equations \eqref{03059} and \eqref{03060} into equation \eqref{03055}, then we obtain price at time $t$ of the zero--coupon bond
\begin{equation}\label{03061}
B_{t,u}(\mathcal{F}_t)=\exp\bigg\{-\tilde{r}_{t+1}-\sum_{\beta=t+1}^{u-1}\big(\tilde{\mu}_{\beta|t}^z(\mathcal{F}_t)\big)_1+\frac{1}{2}\sum_{\alpha=t+1}^{u-1}\sum_{\beta=t+1}^{u-1}\big(\Sigma_{\alpha,\beta|t}(\mathcal{F}_t)\big)_{n+1,n+1}\bigg\}.
\end{equation}
For $k=t+1,\dots,T$, the log price at time $k$ is written in terms of the log price--to--book ratio at time $k$, the log book value growth rates at times $t+1,\dots,k$, and the log book value at time $t$, that is, 
\begin{equation}\label{03062}
\tilde{P}_k=\tilde{m}_k+\tilde{b}_{t+1}+\dots+\tilde{b}_k+\ln(\mathsf{B}_t)
\end{equation}
for $t=0,\dots,k-1$. Let $J_b:=[I_n:0]$ be an $(n\times \tilde{n})$ matrix, whose first block matrix equals $I_n$ and others are zero and $J_m:=[0:I_n]$ be an $(n\times \tilde{n})$ matrix, whose last block matrix equals $I_n$ and others are zero. Those matrices are used to extract the random processes $\tilde{b}_s$ and $\tilde{m}_s$ from the random process $x_s$, i.e., $\tilde{b}_s=J_bx_s$ and $\tilde{m}_s=J_mx_s$ for $s=1,\dots,T$. Therefore, the log price at time $k$ is represented by
\begin{equation}\label{03063}
\tilde{P}_k=K_{k|t}\bar{x}_t^c+\ln(\mathsf{B}_t),
\end{equation}
where the matrix $K_{k|t}$ equals
\begin{equation}\label{ad006}
K_{k|t}:=J_mJ_{k|t}+J_b\sum_{\beta=t+1}^kJ_{\beta|t}.
\end{equation}
Let for $s=t+1,\dots,k$, $\hat{\mu}_{s|t,u}^b(\mathcal{F}_t):=J_b\hat{\mu}_{s|t,u}(\mathcal{F}_t)$, $\hat{\mu}_{s|t,u}^z(\mathcal{F}_t):=J_z\hat{\mu}_{s|t,u}(\mathcal{F}_t)$, and $\hat{\mu}_{s|t}^m(\mathcal{F}_t):=J_m\hat{\mu}_{s|t}(\mathcal{F}_t)$ be expectations of the random processes $\tilde{b}_s$, $z_s$ and $\tilde{m}_s$ under the $(t,u)$--forward probability measure $\mathbb{\hat{P}}_{t,u}$ for given $\mathcal{F}_t$. Hence, conditional on $\mathcal{F}_t$, an expectation of the log price at time $k$ under the $(t,u)$--forward probability measure $\mathbb{\hat{P}}_{t,u}$ equals
\begin{equation}\label{03064}
\hat{\mu}_{k|t,u}^{\tilde{P}}(\mathcal{F}_t):=\mathbb{\hat{E}}_{t,u}\big[\tilde{P}_k\big|\mathcal{F}_t\big]=\hat{\mu}_{k|t,u}^m(\mathcal{F}_t)+\sum_{\beta=t+1}^k\hat{\mu}_{\beta|t,u}^b(\mathcal{F}_t)+\ln(\mathsf{B}_t),
\end{equation}
where $\mathbb{\hat{E}}_{t,u}$ is an expectation under the $(t,u)$--forward probability measure $\mathbb{\hat{P}}_{t,u}$. According to the fact that $J_{s_1|t}\Sigma_t^c(\mathcal{F}_t) J_{s_2|t}'=\Sigma_{s_1,s_2|t}(\mathcal{F}_t)$, conditional on $\mathcal{F}_t$, a covariance matrix of the log price at time $T$ under the $(t,u)$--forward probability measure $\mathbb{\hat{P}}_{t,u}$ is obtained by 
\begin{eqnarray}\label{03065}
\Sigma_{k|t}^{\tilde{P}}(\mathcal{F}_t)&:=&\widehat{\text{Var}}\big[\tilde{P}_k\big|\mathcal{F}_t\big]=\sum_{\alpha=t+1}^k\sum_{\beta=t+1}^kJ_b\Sigma_{\alpha,\beta|t}(\mathcal{F}_t)J_b'+\sum_{\alpha=t+1}^kJ_b\Sigma_{\alpha,k|t}(\mathcal{F}_t)J_m'\nonumber\\
&+&\sum_{\beta=t+1}^kJ_m\Sigma_{k,\beta|t}(\mathcal{F}_t)J_b'+J_m\Sigma_{k,k|t}(\mathcal{F}_t)J_m'.
\end{eqnarray}
Consequently, conditional on $\mathcal{F}_t$, a distribution of the log price at time $k$ is given by
\begin{equation}\label{03066}
\tilde{P}_k~|~\mathcal{F}_t\sim \mathcal{N}\big(\hat{\mu}_{k|t,u}^{\tilde{P}}(\mathcal{F}_t),\Sigma_{k|t}^{\tilde{P}}(\mathcal{F}_t)\big)
\end{equation}
under the $(t,u)$--forward probability measure $\mathbb{\hat{P}}_{t,u}$. Therefore, according to equations \eqref{03064} and \eqref{03065} and Lemma \ref{lem01}, see Technical Annex, price vectors at time $t$ of the Black--Sholes call and put options with strike price vector $K$ and maturity $T$ are given by
\begin{eqnarray}\label{03067}
C_{T|t}(K)&=&\mathbb{\tilde{E}}\bigg[\frac{D_T}{D_t}\Big(P_{T}-K\Big)^+\bigg|\mathcal{F}_t\bigg]=B_{t,T}(\mathcal{F}_t)\mathbb{\hat{E}}_{t,T}\big[\big(P_{T}-K\big)^+\big|\mathcal{F}_t\big]\\
&=&B_{t,T}(\mathcal{F}_t)\bigg(\exp\bigg\{\hat{\mu}_{T|t,T}^{\tilde{P}}(\mathcal{F}_t)+\frac{1}{2}\mathcal{D}\big[\Sigma_{T|t}^{\tilde{P}}(\mathcal{F}_t)\big]\bigg\}\odot\Phi\big(d_{T|t}^1(\mathcal{F}_t)\big)-K\odot\Phi\big(d_{T|t}^2(\mathcal{F}_t)\big)\bigg)\nonumber
\end{eqnarray}
and
\begin{eqnarray}\label{03068}
P_{T|t}(K)&=&\mathbb{\tilde{E}}\bigg[\frac{D_T}{D_t}\Big(K-P_T\Big)^+\bigg|\mathcal{F}_t\bigg]=B_{t,T}(\mathcal{F}_t)\mathbb{\hat{E}}_{t,T}\big[\big(K-P_{T}\big)^+\big|\mathcal{F}_t\big]\\
&=&B_{t,T}(\mathcal{F}_t)\bigg(K\odot\Phi\big(-d_{T|t}^2(\mathcal{F}_t)\big)-\exp\bigg\{\hat{\mu}_{T|t,T}^{\tilde{P}}(\mathcal{F}_t)+\frac{1}{2}\mathcal{D}\big[\Sigma_{T|t}^{\tilde{P}}(\mathcal{F}_t)\big]\bigg\}\odot\Phi\big(-d_{T|t}^1(\mathcal{F}_t)\big)\bigg),\nonumber
\end{eqnarray}
where $d_{T|t}^1(\mathcal{F}_t):=\Big(\hat{\mu}_{T|t,T}^{\tilde{P}}(\mathcal{F}_t)+\mathcal{D}\big[\Sigma_{T|t}^{\tilde{P}}(\mathcal{F}_t)\big]-\ln(K)\Big)\oslash\sqrt{\mathcal{D}\big[\Sigma_{T|t}^{\tilde{P}}(\mathcal{F}_t)\big]}$ and $d_{T|t}^2(\mathcal{F}_t):=d_{T|t}^1(\mathcal{F}_t)-\sqrt{\mathcal{D}\big[\Sigma_{T|t}^{\tilde{P}}(\mathcal{F}_t)\big]}$. 
It should be noted that following the ideas in \citeA{Battulga24e}, one can develop pricing formulas for Margrabe exchange options without default risk.

\section{Life Insurance Products}

Now we consider pricing of some equity--linked life insurance products using the risk--neutral measure. Here we will price segregated fund contracts with guarantees, see \citeA{Hardy01} and unit--linked life insurances with guarantees, see \citeA{Aase94} and \citeA{Moller98}. For discrete--time life insurance products, which cover both of the equity--linked life insurance products, we refer to recent work of \citeA{Battulga24f} We suppose that the stocks represent some funds and an insured receives dividends from the funds. Let $T_x$ be $x$ aged insured's future lifetime random variable, $\mathcal{T}_t^x=\sigma(1_{\{T_x> s\}}:s\in[0,t])$ be $\sigma$--field, which is generated by a death indicator process $1_{\{T_x\leq t\}}$, $F_t$ be an $(n\times 1)$ vector of units of the funds, and $G_t$ be an $(n\times 1)$ vector of amounts of the guarantees, respectively, at time $t$. We assume that the $\sigma$--fields $\mathcal{G}_T$ and $\mathcal{T}_T^x$ are independent, and operational expenses, which are deducted from the funds and withdrawals are omitted from the life insurance products. A common life insurance product in practice is endowment insurance, and combinations of term life insurance and pure endowment insurance lead to various endowment insurances, see \citeA{Aase94}. Thus, it is sufficient to consider only the term life insurance and the pure endowment insurance.  

A $T$--year pure endowment insurance provides payment of a sum insured at the end of the $T$ years only if an insured is alive at the end of $T$ years from the time of policy issue. For the pure endowment insurance, we assume that the sum insured is forming $f(P_T)$ for some Borel function $f: \mathbb{R}_+^n\to \mathbb{R}_+^n$, where $\mathbb{R}_+^n:=\{x\in\mathbb{R}^n|x>0\}$ is the set of $(n\times 1)$ positive real vectors. In this case, the sum insured depends on the random stock price vector at time $T$, and the form of the function $f$ depends on an insurance contract. Choices of $f$ give us different types of life insurance products. For example, for $x,K\in\mathbb{R}_+^n$, $f(x)=i_n$, $f(x)=x$, $f(x)=\max\{x,K\}=[x-K]^++K$, and $f(x)=[K-x]^+$ correspond to simple life insurance, pure unit--linked, unit--linked with guarantee, and segregated fund contract with guarantee, respectively, see \citeA{Aase94}, \citeA{Bowers97}, and \citeA{Hardy01}. As a result, a discounted contingent claim of the $T$--year pure endowment insurance can be represented by the following equation
\begin{equation}\label{03069}
\overline{H}_T:=D_Tf(P_T)1_{\{T_x>T\}}.
\end{equation}
To price the contingent claim, for $t=1,\dots,T$, we define $\sigma$--fields $\mathcal{F}_t^x:=\mathcal{F}_t\vee \mathcal{T}_t^x$, which represents available information for analyst . Since the $\sigma$--fields $\mathcal{G}_T$ and $\mathcal{T}_T^x$ are independent, one can obtain that value at time $t$ of a contingent claim $f(P_T)1_{\{T_x>T\}}$ is given by
\begin{equation}\label{03070}
V_t=\frac{1}{D_t}\mathbb{\tilde{E}}[\overline{H}_T|\mathcal{F}_t^x]=\frac{1}{D_t}\mathbb{\tilde{E}}[D_Tf(P_T)|\mathcal{F}_t]{}_{T-t}p_{x+t},
\end{equation}
where $_tp_x:=\mathbb{P}[T_x>t]$ represents a probability that $x$--aged insured will attain age $x+t$. 

A $T$--year term life insurance is an insurance that provides payment of a sum insured only if death occurs in $T$ years. In contrast to pure endowment insurance, the term life insurance's sum insured depends on time $t$, that is, its sum insured form is $f(P_t)$ because random death occurs at any time in $T$ years. Therefore, a discounted contingent claim of the $T$--term life insurance is given by
\begin{equation}\label{03071}
\overline{H}_T:=D_{K_x+1}f(P_{K_x+1})1_{\{K_x+1\leq T\}}=\sum_{k=0}^{T-1}D_{k+1}f(P_{t+k})1_{\{K_x=k\}},
\end{equation}
where $K_x:=[T_x]$ is the curtate future lifetime random variable of life--aged--$x$. For the contingent claim of the term life insurance, providing a benefit at the end of the year of death, it follows from the fact that $\mathcal{G}_T$ and $\mathcal{T}_T^x$ are independent that a value process at time $t$ of the term insurance is
\begin{equation}\label{03072}
V_t=\frac{1}{D_t}\mathbb{\tilde{E}}[\overline{H}_T|\mathcal{F}_t^x]=\sum_{k=t}^{T-1}\frac{1}{D_t}\mathbb{\tilde{E}}[D_{k+1}f(P_{k+1})|\mathcal{F}_t]{}_{k-t}p_{x+t}q_{x+k}.
\end{equation}
where $_tq_x:=\mathbb{P}[T_x\leq t]$ represents a probability that $x$--aged insured will die within $t$ years.

For the $T$--year term life insurance and $T$--year pure endowment insurance both of which correspond to the segregated fund contract, observe that the sum insured forms are $f(P_k)=F_k\odot\big[L_k-P_k\big]^+$ for $k=1,\dots,T$, where $L_k:=G_k\oslash F_k$. On the other hand, the sum insured forms of the unit--linked life insurance are $f(P_k)=F_k\odot\big[P_k-L_k\big]^++G_k$ for $k=1,\dots,T$. Therefore, from the structure of the sum insureds of the segregated funds and the unit--linked life insurances, one can conclude that to price the life insurance products it is sufficient to consider European call and put options with strike price $L_k$ and maturity $k$ for $k=t+1,\dots,T$. 

Since conditional distributions of the stock prices at times $k=t+1,\dots,T$ are given by \eqref{03066}, similarly to equations \eqref{03067} and \eqref{03068}, one can obtain that for $k=t+1,\dots,T$,
\begin{eqnarray}\label{03074}
&&C_{k|t}(L_k)=\mathbb{\tilde{E}}\bigg[\frac{D_k}{D_t}\Big(P_k-L_k\Big)^+\bigg|\mathcal{F}_t\bigg]\\
&&=B_{t,k}(\mathcal{F}_t)\bigg(\exp\bigg\{\hat{\mu}_{k|t,k}^{\tilde{P}}(\mathcal{F}_t)+\frac{1}{2}\mathcal{D}\big[\Sigma_{k|t}^{\tilde{P}}(\mathcal{F}_t)\big]\bigg\}\odot\Phi\big(d_{k|t}^1(\mathcal{F}_t)\big)-L_k\odot\Phi\big(d_{k|t}^2(\mathcal{F}_t)\big)\bigg)\nonumber
\end{eqnarray}
and
\begin{eqnarray}\label{03075}
&&P_{k|t}(L_k)=\mathbb{\tilde{E}}\bigg[\frac{D_k}{D_t}\Big(L_k-P_k\Big)^+\bigg|\mathcal{F}_t\bigg]\\
&&=B_{t,k}(\mathcal{F}_t)\bigg(L_k\odot\Phi\big(-d_{k|t}^2(\mathcal{F}_t)\big)-\exp\bigg\{\hat{\mu}_{k|t,k}^{\tilde{P}}(\mathcal{F}_t)+\frac{1}{2}\mathcal{D}\big[\Sigma_{k|t}^{\tilde{P}}(\mathcal{F}_t)\big]\bigg\}\odot\Phi\big(-d_{k|t}^1(\mathcal{F}_t)\big)\bigg),\nonumber
\end{eqnarray}
where $d_{k|t}^1(\mathcal{F}_t):=\Big(\hat{\mu}_{k|t,k}^{\tilde{P}}(\mathcal{F}_t)+\mathcal{D}\big[\Sigma_{k|t}^{\tilde{P}}(\mathcal{F}_t)\big]-\ln(L_k)\Big)\oslash\sqrt{\mathcal{D}\big[\Sigma_{k|t}^{\tilde{P}}(\mathcal{F}_t)\big]}$ and $d_{k|t}^2(\mathcal{F}_t):=d_{k|t}^1(\mathcal{F}_t)-\sqrt{\mathcal{D}\big[\Sigma_{k|t}^{\tilde{P}}(\mathcal{F}_t)\big]}$.

Consequently, from equations \eqref{03074} and \eqref{03075}, net single premiums of the $T$--year life insurance products without withdrawal and operational expenses, providing a benefit at the end of the year of death (term life insurance) or the end of the year $T$ (pure endowment insurance) are given by
\begin{itemize}
\item[1.] for the $T$--year guaranteed term life insurance, corresponding to segregated fund contract, it holds
\begin{equation}\label{03076}
\lcterm{S}{x+t}{T-t}=\sum_{k=t}^{T-1}F_{k+1}\odot P_{k+1|t}(L_{k+1}){}_{k-t}p_{x+t}q_{x+k};
\end{equation}
\item[2.] for the $T$--year guaranteed pure endowment insurance, corresponding to segregated fund contract, it holds
\begin{equation}\label{03077}
\lcend{S}{x+t}{T-t}=F_T\odot P_{T|t}(L_T){}_{T-t}p_{x+t};
\end{equation}
\item[3.] for the $T$--year guaranteed unit--linked term life insurance, it holds
\begin{equation}\label{03078}
\lcterm{U}{x+t}{T-t}=\sum_{k=t}^{T-1}\big[F_{k+1}\odot C_{k+1|t}(L_{k+1})+B_{t,k+1}(\mathcal{F}_t)G_{k+1}\big]{}_{k-t}p_{x+t}q_{x+k};
\end{equation}
\item[4.] and for the $T$--year guaranteed unit--linked pure endowment insurance, it holds
\begin{equation}\label{03079}
\lcend{U}{x+t}{T-t}=\big[F_T\odot C_{T|t}(L_T)+B_{t,T}(\mathcal{F}_t)G_T\big]{}_{T-t}p_{x+t}.
\end{equation}
\end{itemize}

\section{Locally Risk--Minimizing Strategy}

By introducing the concept of mean--self--financing, \citeA{Follmer86} extended the concept of the complete market into the incomplete market. If a discounted cumulative cost process is a martingale, then a portfolio plan is called mean--self--financing. In a discrete--time case, \citeA{Follmer89} developed a locally risk--minimizing strategy and obtained a recurrence formula for optimal strategy. According to \citeA{Schal94} (see also \citeA{Follmer04}), under a martingale probability measure the locally risk--minimizing strategy and remaining conditional risk--minimizing strategy are the same. Therefore, in this section, we will consider locally risk--minimizing strategies, which correspond to the Black--Scholes call and put options given in Section 3 and the equity--linked life insurance products given in Section 4. In the insurance industry, for continuous--time unit--linked term life and pure endowment insurances with guarantee, locally risk--minimizing strategies are obtained by \citeA{Moller98}. Recently, for discrete--time equity--linked life insurance products, \citeA{Battulga24f} obtained locally risk--minimizing strategies. 

To simplify notations we define: for $t=1,\dots,T$, $\overline{P}_t:=(\overline{P}_{1,t},\dots,\overline{P}_{n,t})'$ is a discounted stock price process at time $t$, $\overline{d}_t:=(\overline{d}_{1,t},\dots,\overline{d}_{n,t})'$ is a discounted dividend payment process at time $t$, and $\Delta \overline{P}_t:=\overline{P}_t-\overline{P}_{t-1}$ and $\Delta \overline{d}_t:=\overline{d}_t-\overline{d}_{t-1}$ are difference processes at time $t$ of the discounted stock price and dividend processes, respectively, where $\overline{P}_{i,t}:=D_tP_{i,t}$ and $\overline{d}_{i,t}:=D_td_{i,t}$ are discounted stock price process and discounted dividend payment process, respectively, at time $t$ of $i$--th stock. For $i=1,\dots,n$, let $h_{i,t}$ be a proper number of shares at time $t$ and $h_{i,t}^0$ be a proper amount of cash (risk--free bond) at time $t$, which are required to successfully hedge $i$--th contingent claim $H_{i,T}$, and $\overline{H}_{i,T}$ be a discounted contingent claim, where we assume that the contingent claim $\overline{H}_{i,T}$ is square--integrable under the risk--neutral probability measure. 

To obtain locally risk--minimizing strategy ($h_i^0, h_i$), corresponding to the $i$--th contingent claim $H_{i,T}$, we follow \citeA{Follmer04} and \citeA{Follmer89}. Let $\tilde{\mathbb{P}}^\circ$ be a martingale probability measure satisfying 
\begin{equation}\label{05.001}
\tilde{\mathbb{E}}^\circ[\overline{P}_{t+1}+\overline{d}_{t+1}|\mathcal{G}_t]=\overline{P}_t+\overline{d}_t,
\end{equation}
where $\tilde{\mathbb{E}}^\circ$ is an expectation under the martingale probability measure $\tilde{\mathbb{P}}^\circ$. Note that $\tilde{\mathbb{P}}^\circ$ can be any probability measure. For example, one may choose the probability measure $\tilde{\mathbb{P}}^\circ$ by the risk--neutral probability measure $\tilde{\mathbb{P}}$. In this case, it is very difficult to obtain the locally risk--minimizing strategy ($h_i^0, h_i$). Discounted portfolio value at time $t$, corresponding to the $i$--th contingent claim $H_{i,T}$ is given by
\begin{equation}\label{05.002}
\overline{V}_{i,t}=h_{i,t}'(\overline{P}_t+\overline{d}_t)+\overline{h}_{i,t}^0,~~~i=1,\dots,n,~t=1,\dots,T
\end{equation}
Note that $h_{i,t}$ is a predictable process, which means its value is known at time $(t-1)$, while for the process $h_{i,t}^0$, its value is only known at time $t$,
We suppose that the final discounted portfolio value replicates the $i$--th discounted contingent claim, that is,
\begin{equation}\label{05.004}
\overline{V}_{i,T}=\overline{H}_{i,T}.
\end{equation}
A discounted cumulative cost process is defined by
\begin{equation}\label{05.005}
\overline{C}_{i,t}:=\overline{V}_{i,t}-\sum_{j=1}^th_{i,j}'(\Delta\overline{P}_j+\Delta\overline{d}_j)
\end{equation}
and $C_{i,0}=V_{i,0}$. To obtain locally risk--minimizing strategy ($h_i^0, h_i$), corresponding to the $i$--th contingent claim $H_{i,T}$, we need to minimize a conditional local risk, which is defined by 
\begin{equation}\label{05.006}
R_{i,t}:=\tilde{\mathbb{E}}^\circ\big[(\overline{C}_{i,t+1}-\overline{C}_{i,t})^2\big|\mathcal{F}_t^x\big]=\tilde{\mathbb{E}}^\circ\big[(\overline{V}_{i,t+1}-\overline{V}_{i,t}-h_{i,t+1}'(\Delta\overline{P}_{t+1}+\Delta\overline{d}_{t+1}))^2\big|\mathcal{F}_t^x\big].
\end{equation}
The above optimization problem corresponds to individual contingent claim and it does not take into account correlations between the contingent claims. For this reason, instead of the above optimization problem, we consider the following optimization problem 
\begin{equation}\label{05.007}
R_t:=\sum_{i=1}^n\tilde{\mathbb{E}}^\circ\big[(\overline{C}_{i,t+1}-\overline{C}_{i,t})^2\big|\mathcal{F}_t^x\big]\longrightarrow \text{min}.
\end{equation}
In order to solve the above optimization problem, we define the following vectors and matrix: $\overline{V}_t:=(\overline{V}_{1,t},\dots,\overline{V}_{n,t})'$ is an $(n\times 1)$ vector of discounted portfolio value process at time $t$, $\overline{H}_T:=(\overline{H}_{1,T},\dots,\overline{H}_{n,T})'$ is an $(n\times 1)$ vector of discounted contingent claim at time $T$, and $h_t:=[h_{1,t}:\dots:h_{n,t}]$ is an $(n\times n)$ proper number of shares matrix at time $t$. Then, the optimization problem becomes
\begin{equation}\label{05.008}
R_t:=\tilde{\mathbb{E}}^\circ\Big[\Big(\overline{V}_{t+1}-\overline{V}_t-h_{t+1}'(\Delta\overline{P}_{t+1}+\Delta\overline{d}_{t+1})\Big)'\Big(\overline{V}_{t+1}-\overline{V}_t-h_{t+1}'(\Delta\overline{P}_{t+1}+\Delta\overline{d}_{t+1})\Big)\Big|\mathcal{F}_t^x\Big]\longrightarrow \text{min}.
\end{equation}
To solve recursively the optimization problem with respect to the parameters $\overline{V}_t$ and $h_t$, we start $t=T-1$ with $\overline{V}_T=\overline{H}_T$. Partial derivatives from the objective function $R_t$ with respect to parameters $\overline{V}_t$ and $h_t$ are
\begin{equation}\label{05.009}
\frac{\partial R_t}{\partial\overline{V}_{t}}=-\tilde{\mathbb{E}}^\circ\Big[\overline{V}_{t+1}-\overline{V}_t-h_{t+1}'(\Delta\overline{P}_{t+1}+\Delta\overline{d}_{t+1})\Big|\mathcal{F}_t^x\Big]
\end{equation}
and
\begin{equation}\label{05.010}
\frac{\partial R_t}{\partial h_{t+1}}=2\tilde{\mathbb{E}}^\circ\Big[(\Delta\overline{P}_{t+1}+\Delta\overline{d}_{t+1})(\Delta\overline{P}_{t+1}+\Delta\overline{d}_{t+1})'\Big|\mathcal{F}_t^x\Big]h_{t+1}-2\tilde{\mathbb{E}}^\circ\Big[(\Delta\overline{P}_{t+1}+\Delta\overline{d}_{t+1})(\overline{V}_{t+1}-\overline{V}_t)'\Big|\mathcal{F}_t^x\Big],
\end{equation}
respectively. Since $\Delta\overline{P}_{t+1}+\Delta\overline{d}_{t+1}$ is a martingale difference, we have that
\begin{equation}\label{05.011}
\overline{V}_t=\tilde{\mathbb{E}}^\circ\big[\overline{V}_{t+1}\big|\mathcal{F}_t^x\big]~~~\text{and}~~~h_{t+1}=\overline{\Omega}_{t+1}^{-1}\overline{\Lambda}_{t+1}
\end{equation}
for $t=0,\dots,T-1$, where $\overline{\Omega}_{t+1}:=\tilde{\mathbb{E}}^\circ\big[(\Delta\overline{P}_{t+1}+\Delta\overline{d}_{t+1})(\Delta\overline{P}_{t+1}+\Delta\overline{d}_{t+1})'\big|\mathcal{F}_{t}^x\big]$ is an $(n\times n)$ random matrix and $\overline{\Lambda}_{t+1}:=\tilde{\mathbb{E}}^\circ\big[(\Delta\overline{P}_{t+1}+\Delta\overline{d}_{t+1})\overline{V}_{t+1}'\big|\mathcal{F}_{t}^x\big]$  is an $(n\times n)$ random matrix. As $\overline{V}_T=\overline{H}_T$, by equation \eqref{05.011} and tower property of conditional expectation, it can be shown that $\overline{\Lambda}_{t+1}:=\tilde{\mathbb{E}}^\circ\big[(\Delta\overline{P}_{t+1}+\Delta\overline{d}_{t+1})\overline{H}_T'\big|\mathcal{F}_{t}^x\big]$. Consequently, due to equation \eqref{05.002}, under the martingale probability measure $\mathbb{\tilde{P}}^\circ$, the locally risk--minimizing strategy ($h^0, h$) is given by the following equations:
\begin{equation}\label{05.012}
h_{t+1}=\overline{\Omega}_{t+1}^{-1}\overline{\Lambda}_{t+1}~~~\text{and}~~~h_{t+1}^0=V_{t+1}-h_{t+1}'(P_{t+1}+d_{t+1})
\end{equation}
for $t=0,\dots,T-1$ and $h_0^0=V_0-h_1'(P_0+d_0)$, where $V_{t+1}:=\frac{1}{D_{t+1}}\mathbb{\tilde{E}}^\circ[\overline{H}_T|\mathcal{F}_{t+1}^x]$ is a value process of the contingent claim $H_T$. If the
contingent claim $H_T$ is generated by stock price process $P_t$ and dividend process $d_t$ for $t=1,\dots,T$, then the process $h_t^0$
becomes predictable, see \citeA{Follmer89}. 

\subsection{Martingale and Forward Probability Measures}

By equation \eqref{03004}, martingale condition \eqref{05.001} is equivalent to the following condition
\begin{equation}\label{05.014}
\tilde{\mathbb{E}}^\circ\Big[\exp\Big\{u_t-\Big(\tilde{r}_ti_{n}-C_{k}\psi_t+(I_n-G_{t-1})(\tilde{d}_{t-1}-\tilde{P}_{t-1})+G_{t-1}^{-1}h_{t-1}\Big)\Big\}\Big|\mathcal{G}_{t-1}\Big]=i_{n}.
\end{equation}
Again, to obtain locally risk--minimizing strategies for the options and life insurance products, we use the unique optimal Girsanov kernel process $\theta_t$, which minimizes the variance of a state price density process and relative entropy. According to \citeA{Battulga23a}, the optimal kernel process $\theta_t$ is given by
\begin{equation}\label{05.015}
\theta_t=\Theta_t \bigg(\tilde{r}_ti_n-C_{k}\psi_t+\big(I_n-G_{t-1}^{-1}\big)(\tilde{d}_{t-1}-\tilde{P}_{t-1})+G_{t-1}^{-1}h_{t-1}-\frac{1}{2}\mathcal{D}[\Sigma_{uu}]\bigg).
\end{equation}
Consequently, system \eqref{03017} can be written by for $t=1,\dots,T$,
\begin{equation}\label{05.016}
\begin{cases}
\tilde{b}_t=\tilde{\nu}_{b,t}^\circ-\Psi_{b,t}\tilde{m}_t^*+E_t\tilde{z}_{t-1}^*+G_t\tilde{u}_t^\circ\\
z_t^*=\tilde{\nu}_{z,t}^{*\circ}+\tilde{A}^*z_{t-1}^*+\Psi_{z,t}\tilde{m}_{t-1}^*+\tilde{v}_t^{*\circ}\\
\tilde{m}_t^*=\tilde{\nu}_{m,t}^*+C\tilde{m}_{t-1}^*+w_t^{*\circ}
\end{cases}
\end{equation}
under the martingale probability measure $\tilde{\mathbb{P}}^\circ$, where $\tilde{\nu}_{b,t}^\circ:=G_t\big((I_n-G_{t-1}^{-1})\tilde{d}_{t-1}+G_{t-1}^{-1}(h_{t-1}+\ln(B_{t-1}))-\frac{1}{2}\mathcal{D}[\Sigma_{uu}]\big)+(I_n-G_t)\tilde{\Delta}_t-G_t\ln(B_{t-1})-h_t$ is an $(n\times 1)$ intercept process of the log book value growth rate process $\tilde{b}_t$, $\tilde{\nu}_{z,t}^{*\circ}:=\big((\tilde{\nu}_{z,t}^{\circ})',0,\dots,0\big)'$ with $\tilde{\nu}_{z,t}^{\circ}:=C_z\psi_t+\Sigma_{vu}\Sigma_{uu}^{-1}\big((I_n-G_{t-1}^{-1})(\tilde{d}_{t-1}-\ln(B_{t-1}))\big)-C_k\psi_t+G_{t-1}^{-1}h_{t-1}-\frac{1}{2}\mathcal{D}[\Sigma_{uu}]$ is an $(\ell p\times 1)$ intercept process of the economic variables process $z_t^*$, and
\begin{equation}\label{ad013}
\Psi_{z,t}:=\begin{bmatrix}
\Sigma_{vu}\Sigma_{uu} & 0\\
0 & 0
\end{bmatrix}
\end{equation}
is an $(\ell p \times 2n)$ matrix. A joint distribution of a random vector $\tilde{\xi}^{*\circ}:=\big((\tilde{\xi}_1^{*\circ})',\dots,(\tilde{\xi}_T^{*\circ})'\big)$ with $\tilde{\xi}_t^{*\circ}:=\big((\tilde{u}_t^\circ)',(\tilde{v}_t^{*\circ})',(\tilde{w}_t^{*\circ})'\big)'$ is given by
\begin{equation}\label{ad014}
\tilde{\xi}^{*\circ}\sim \mathcal{N}(0,I_T\otimes\Sigma_{\xi\xi})
\end{equation}
under the martingale probability measure $\tilde{\mathbb{P}}^\circ$. From the first line of the above system, one may be conclude that the log price at time $t$ equals
\begin{eqnarray}\label{ad015}
\tilde{P}_t&=&G_t\bigg(\tilde{r}_ti_n-\frac{1}{2}\mathcal{D}[\Sigma_{uu}]+(I_n-G_{t-1}^{-1})\tilde{d}_{t-1}+G_{t-1}^{-1}\big(\tilde{P}_{t-1}+h_{t-1}\big)\bigg)\nonumber\\
&+&(I_n-G_t)\tilde{d}_t-h_t+G_t\tilde{u}_t^*
\end{eqnarray}
under the martingale probability measure $\tilde{\mathbb{P}}^\circ$.

One may be write the system \eqref{05.016} in VAR(1) form
\begin{equation}\label{ad016}
Q_{0,t}x_t^*=\tilde{\nu}_t^{*\circ}+Q_{1,t}^\circ x_{t-1}^*+\mathsf{G}_t^*\tilde{\xi}^{*\circ}
\end{equation}
under the martingale probability measure $\tilde{\mathbb{P}}^\circ$, where $\tilde{\nu}_t^\circ:=\big((\tilde{\nu}_{b,t}^\circ)',(\tilde{\nu}_{z,t}^{*\circ})',(\tilde{\nu}_{m,t}^{*\circ})'\big)'$ an intercept process of the VAR(1) process $x_t^*$, and
\begin{equation}\label{ad017}
Q_{1,t}^\circ:=\begin{bmatrix}
0 & E_t & 0\\
0 & \tilde{A}_t^* & \Psi_{z,t}\\
0 & 0 & C
\end{bmatrix}.
\end{equation}
By repeating equation \eqref{ad016}, one gets that for $s=t+1,\dots,T$,
\begin{equation}\label{05.021}
x_s^*=\Pi_{t,s}^{*\circ} x_t^*+\sum_{\beta=t+1}^s\Pi_{\beta,s}^{*\circ}\tilde{\nu}_\beta^{*\circ}+\sum_{\beta=t+1}^s\Pi_{\beta,s}^{*\circ}\mathsf{G}_\beta^*\tilde{\xi}_\beta^{*\circ},
\end{equation}
where the coefficient matrices are for $\beta=t$,
\begin{equation}\label{05.022}
\Pi_{\beta,s}^{*\circ}:=\prod_{\alpha=\beta+1}^sQ_{0,\alpha}^{-1}Q_{1,\alpha}^\circ,
\end{equation}
for $\beta=t+1,\dots,s-1$,
\begin{eqnarray}\label{05.023}
\Pi_{\beta,s}^{*\circ}:=\Bigg(\prod_{\alpha=\beta+1}^sQ_{0,\alpha}^{-1}Q_{1,\alpha}^\circ\Bigg)Q_{0,\beta}^{-1},
\end{eqnarray}
and for $\beta=s$,
\begin{equation}\label{05.024}
\Pi_{\beta,s}^{*\circ}:=\Pi_{s,s}^*
\end{equation}
Consequently, we have that
\begin{equation}\label{ad018}
x_s=\mathsf{J}_x\Pi_{t,s}^{*\circ} x_t^*+\sum_{\beta=t+1}^s\Pi_{\beta,s}^{\circ}\tilde{\nu}_\beta^\circ+\sum_{\beta=t+1}^s\Pi_{\beta,s}^{\circ}\mathsf{G}_\beta\tilde{\xi}_\beta^{\circ},
\end{equation}
under the martingale probability measure $\tilde{\mathbb{P}}^\circ$, where $\Pi_{\beta,s}^\circ:=\mathsf{J}_x\Pi_{\beta,s}^{*\circ}\mathsf{J}_x'$ is an $(\tilde{n}\times \tilde{n})$ matrix, $\tilde{\nu}_\beta^\circ:=\mathsf{J}_x\tilde{\nu}_\beta^{*\circ}$ is an $(\tilde{n}\times 1)$ vector, and $\tilde{\xi}_\beta^\circ:=\mathsf{J}_x\tilde{\xi}_\beta^{*\circ}$ is an $(\tilde{n}\times 1)$ vector. Thus, we have that
\begin{equation}\label{ad019}
\tilde{\mu}_{s|t}^\circ(\mathcal{G}_t):=\mathbb{\tilde{E}}^\circ[x_s|\mathcal{G}_t]=\mathsf{J}_x\Pi_{t,s}^{*\circ} x_t^*+\sum_{\beta=t+1}^s\Pi_{\beta,s}^{\circ}\tilde{\nu}_\beta^\circ,~~~s=t+1,\dots,T
\end{equation}
and
\begin{equation}\label{ad020}
\Sigma_{s_1,s_2|t}^\circ(\mathcal{G}_t):=\widetilde{\text{Cov}}^\circ[x_{s_1},x_{s_2}|\mathcal{G}_t]=\sum_{\beta=t+1}^{s_1\wedge s_2}\Pi_{\beta,s}^{\circ}\mathsf{G}_\beta\Sigma_{\xi\xi}\mathsf{G}_\beta(\Pi_{\beta,s}^{\circ})'.
\end{equation}
Consequently, conditional on $\mathcal{G}_t$, a distribution of the random vector $\bar{x}_t^c$ is
\begin{equation}\label{03087}
\bar{x}_t^c~|~\mathcal{G}_t\sim \mathcal{N}\big(\tilde{\mu}_t^{c\circ}(\mathcal{G}_t),\Sigma_t^{c\circ}(\mathcal{G}_t)\big)
\end{equation}
under the martingale probability measure $\tilde{\mathbb{P}}^\circ$, where $\tilde{\mu}_{s|t}^{c\circ}(\mathcal{G}_t):=\big((\tilde{\mu}_{s|t}^\circ(\mathcal{G}_t))',\dots,(\tilde{\mu}_{s|t}^\circ(\mathcal{G}_t))'\big)'$ is expectation and $\Sigma_t^{c\circ}(\mathcal{G}_t):=\big(\Sigma_{s_1,s_2|t}^\circ(\mathcal{G}_t)\big)_{s_1,s_2=t+1}^T$ is covariance matrix of the random vector $\bar{x}_t^c$ for given $\mathcal{G}_t$. 

A $(t,u)$--forward probability measure, which is originated from the martingale measure $\tilde{\mathbb{P}}^\circ$ is defined by
\begin{equation}\label{05.028}
\mathbb{\hat{P}}_{t,u}^\circ\big[A\big|\mathcal{G}_t\big]:=\frac{1}{D_tB_{t,u}^\circ(\mathcal{G}_t)}\int_AD_u\mathbb{\tilde{P}}^\circ\big[\omega|\mathcal{G}_t\big]~~~\text{for all}~A\in \mathcal{G}_T,
\end{equation}
where $B_{t,u}^\circ(\mathcal{G}_t):=\frac{1}{D_t}\mathbb{\tilde{E}}^\circ\big[D_u\big|\mathcal{G}_t\big]$. Then, similarly to Subsection 3.2, it can be shown that
\begin{equation}\label{05.029}
B_{t,u}^\circ(\mathcal{G}_t)=\exp\bigg\{-\tilde{r}_{t+1}-\gamma_{t,u}'\tilde{\mu}_t^{c\circ}(\mathcal{G}_t)+\frac{1}{2}\gamma_{t,u}'\Sigma_t^{c\circ}(\mathcal{G}_t)\gamma_{t,u}\bigg\}
\end{equation}
and conditional on the information $\mathcal{G}_t$, a joint distribution of the random vector $\bar{x}_t^c$ is given by
\begin{equation}\label{05.030}
\bar{x}_t^c~|~\mathcal{G}_t\sim \mathcal{N}\big(\hat{\mu}_{t,u}^{\circ c}(\mathcal{G}_t),\Sigma_t^{\circ c}(\mathcal{G}_t)\big),~~~t=0,\dots,T-1
\end{equation}
under the $(t,u)$--forward probability measure $\hat{\mathbb{P}}_{t,u}^\circ$, where $\hat{\mu}_{t,u}^{\circ c}(\mathcal{G}_t):=\tilde{\mu}_t^{\circ c}(\mathcal{G}_t)-\Sigma_t^{\circ c}(\mathcal{G}_t)\gamma_{t,u}$ and $\Sigma_t^{\circ c}(\mathcal{G}_t)$ are conditional expectation and conditional covariance matrix, respectively, of the random vector $\bar{x}_t^c$ for given $\mathcal{G}_t$. Similarly as equation \eqref{03058}, $(s-t)$--th block vector of the conditional expectation $\hat{\mu}_{t,u}^{c\circ}(\mathcal{G}_t)$ is given by
\begin{equation}\label{03088}
\hat{\mu}_{s|t,u}^\circ(\mathcal{G}_t):=J_{s|t}\hat{\mu}_{t,u}^{c\circ}(\mathcal{G}_t)=\tilde{\mu}_{s|t}^\circ(\mathcal{G}_t)-\sum_{\beta=t+1}^{u-1}\big(\Sigma_{s,\beta|t}^\circ(\mathcal{G}_t)\big)_{n+1}.
\end{equation}
Consequently, it follows from equations \eqref{ad018} and \eqref{03087} that
for $s=t+1,\dots,T$,
\begin{eqnarray}\label{03089}
x_s\overset{d}{=}\mathsf{J}_x\Pi_{t,s}^{*\circ}x_t^*+\sum_{\beta=t+1}^s\Pi_{\beta,s}^\circ\tilde{\nu}_\beta^\circ-\sum_{\beta=t+1}^{u-1}\big(\Sigma_{s,\beta|t}^\circ(\mathcal{G}_t)\big)_{n+1}+\sum_{\beta=t+1}^s\Pi_{\beta,s}^\circ\mathsf{G}_\beta\hat{\xi}_\beta^\circ
\end{eqnarray}
under the $(t,u)$--forward probability measure $\mathbb{\hat{P}}_{t,u}^\circ$, where $d$ means equal distribution and the random vector $\hat{\xi}^{c\circ}:=\big((\hat{\xi}_{1}^\circ)',\dots,(\hat{\xi}_T^\circ)'\big)'$ follows multivariate normal distribution with zero mean and covariance matrix $I_T\otimes \Sigma_{\xi\xi}$, i.e,
\begin{equation}\label{03090}
\hat{\xi}^{c\circ}\sim \mathcal{N}(0,I_T\otimes \Sigma_{\xi\xi})
\end{equation}
under the $(t,u)$--forward probability measure $\hat{\mathbb{P}}_{t,u}^\circ$. By equation \eqref{ad020}, it can be shown that
\begin{equation}\label{03091}
\sum_{\beta=t+1}^{u-1}\big(\Sigma_{s,\beta|t}^\circ(\mathcal{G}_t)\big)_{n+1}=\sum_{\beta=t+1}^s\Pi_{\beta,s}^\circ\mathsf{G}_\beta\hat{c}_{\beta|t,u}^\circ,
\end{equation}
where $\hat{c}_{\beta|t,u}^\circ:=\sum_{\alpha=t+1}^{u-1}\big(\Sigma_{\xi\xi}\mathsf{G}_\beta(\Pi_{\beta,\alpha}^\circ)'\big)_{n+1}$ is an $(\tilde{n}\times 1)$ vector. Therefore, we have that
\begin{eqnarray}\label{03092}
x_s\overset{d}{=}\mathsf{J}_x\Pi_{t,s}^{*y\circ}y_t^*+\mathsf{J}_x\Pi_{t,s}^{*m\circ}\tilde{m}_t^*+\sum_{\beta=t+1}^s\Pi_{\beta,s}^\circ\hat{\nu}_\beta^\circ+\sum_{\beta=t+1}^s\Pi_{\beta,s}^\circ\mathsf{G}_\beta\hat{\xi}_\beta^\circ
\end{eqnarray}
under the $(t,u)$--forward probability measure $\hat{\mathbb{P}}_{t,u}^\circ$, where $\Pi_{t,s}^{*\circ}=\big[\Pi_{t,s}^{*y\circ}:\Pi_{t,s}^{*m\circ}\big]$ is a partition, corresponding to the random vectors $y_t^*$ and $\tilde{m}_t^*$ of the matrix $\Pi_{t,s}^{*\circ}$ and $\hat{\nu}_\beta^\circ:=\tilde{\nu}_\beta^\circ-\mathsf{G}_\beta\hat{c}_{\beta|t,u}^\circ$ is an intercept process of the process $x_\beta$ under the $(t,u)$--forward probability measure $\hat{\mathbb{P}}_{t,u}^\circ$, see below. As a result, by comparing equations \eqref{ad018} and \eqref{03092} and using a fact that $J_b\mathsf{G}_t=G_tJ_b$, $J_z\mathsf{G}_t=J_z$, $J_m\mathsf{G}_t=J_m$, and $J_m\hat{c}_{t|t,u}=0$, one can conclude that the log price process $\tilde{P}_t$ is given by
\begin{equation}\label{03093}
\tilde{P}_t=G_t\bigg(\tilde{P}_{t-1}-\tilde{d}_t+\tilde{r}_t i_n-\frac{1}{2}\mathcal{D}[\Sigma_{uu}]-J_b\hat{c}_{t|t,u}^\circ\bigg)+\tilde{d}_t-h_t+G_t\hat{u}_t^\circ
\end{equation}
and system \eqref{03017} becomes 
\begin{equation}\label{03094}
\begin{cases}
\tilde{b}_t=\hat{\nu}_{b,t}^\circ+\Psi_{b,t}\tilde{m}_t^*+E_tz_{t-1}^*+\hat{u}_t^\circ\\
z_t=\hat{\nu}_{z,t}^\circ+\tilde{A}_tz_{t-1}^*+\Sigma_{vu}\Sigma_{uu}^{-1}\tilde{m}_{t-1}+\hat{v}_t^\circ\\
\tilde{m}_t=\nu_{m,t}+\tilde{m}_{t-1}+\hat{w}_t^\circ
\end{cases}~~~\text{for}~t=1,\dots,T
\end{equation}
under the $(t,u)$--forward probability measure $\mathbb{\hat{P}}_{t,u}^\circ$, where $\hat{\nu}_{b,t}^\circ:=\tilde{\nu}_{b,t}^\circ-G_tJ_b\hat{c}_{t|t,u}^\circ$ is an $(n\times 1)$ intercept process of the log book value growth rate process and $\hat{\nu}_{z,t}^\circ:=\tilde{\nu}_{z,t}^\circ-J_z\hat{c}_{t|t,u}^\circ$ is an $(\ell\times 1)$ intercept process of the economic variables process. Again, the last line of system \eqref{03094} does not change because the white noise processes $\eta_t=(u_t',v_t')'$ and $w_t$ and the initial price--to--book ration $\tilde{m}_0$ are independent under the real probability measure $\mathbb{P}$. Also, the random white noise processes $(\hat{u}_t',\hat{v}_t')'$ and $\hat{w}_t$ and the initial price--to--book ration $\tilde{m}_0$ are independent under the $(t,u)$--forward probability measure $\mathbb{\hat{P}}_{t,u}^\circ$ and the distribution of the log price--to--book ratio $\tilde{m}_t$ is same for all the probability measures.

By replacing $\Pi_{\beta,s}^*$, $\Pi_{\beta,s}$, and $\tilde{\nu}_\beta$ in Subsections 3.2 and 3.3 with $\Pi_{\beta,s}^{*\circ}$, $\Pi_{\beta,s}^\circ$, and $\tilde{\nu}_\beta^\circ$ and $\hat{\nu}_\beta^\circ$, conditional on $\mathcal{F}_t$, one obtains distributions, corresponding to equation \eqref{03045}, \eqref{03051}, and \eqref{03056} of the price--to--book ratio process at time $t$, $\tilde{m}_t^*$ and random vector $\bar{x}_t^c$ under the martingale probability measure $\mathbb{\tilde{P}}^\circ$ and $(t,u)$--forward probability measure $\mathbb{\hat{P}}_{t,u}^\circ$, namely
\begin{equation}\label{ad021}
\tilde{m}_t^*~|~\mathcal{F}_t\sim \mathcal{N}\Big(\tilde{\tilde{m}}_{t|t}^{*\circ},\Sigma^{\circ}(\tilde{m}_t^*|t)\Big),~~~t=0,\dots,T-1
\end{equation}
under the martingale probability measure $\tilde{\mathbb{P}}^\circ$,
\begin{equation}\label{ad022}
\tilde{m}_t^*~|~\mathcal{F}_t\sim \mathcal{N}\Big(\hat{\tilde{m}}_{t|t,u}^{*\circ},\Sigma^{\circ}(\tilde{m}_t^*|t)\Big),~~~t=0,\dots,T-1
\end{equation}
under the $(t,u)$--forward probability measure $\hat{\mathbb{P}}_{t,u}^\circ$,
\begin{equation}\label{ad023}
\bar{x}_t^c~|~\mathcal{F}_t\sim \mathcal{N}\Big(\tilde{\mu}_t^{c\circ}(\mathcal{F}_t),\Sigma_t^{c\circ}(\mathcal{F}_t)\Big),~~~t=0,\dots,T-1
\end{equation}
under the martingale probability measure $\tilde{\mathbb{P}}^\circ$,
and
\begin{equation}\label{ad024}
\bar{x}_t^c~|~\mathcal{F}_t\sim \mathcal{N}\Big(\hat{\mu}_{t,u}^{c\circ}(\mathcal{F}_t),\Sigma_t^{c\circ}(\mathcal{F}_t)\Big),~~~t=0,\dots,T-1
\end{equation}
under the $(t,u)$--forward probability measure $\hat{\mathbb{P}}_{t,u}^\circ$, where for the two conditional expectations of the random vector $\bar{x}_t^c$, it holds $\hat{\mu}_t^{c\circ}(\mathcal{F}_t):=\tilde{\mu}_t^{c\circ}(\mathcal{F}_t)-\Sigma_t^{c\circ}(\mathcal{F}_t)\gamma_{t,u}$. Similarly to \eqref{05.029}, from equation \eqref{ad023}, conditional on $\mathcal{F}_t$, one obtains bond price at time $t$ under the martingale probability measure $\tilde{\mathbb{P}}^\circ$
\begin{equation}\label{ad037}
B_{t,u}^\circ(\mathcal{F}_t)=\exp\bigg\{-\tilde{r}_{t+1}-\gamma_{t,u}'\tilde{\mu}_t^{c\circ}(\mathcal{F}_t)+\frac{1}{2}\gamma_{t,u}'\Sigma_t^{c\circ}(\mathcal{F}_t)\gamma_{t,u}\bigg\}.
\end{equation}

\subsection{Parameters of Locally Risk--Minimizing Strategy}

Now we consider parameters $\overline{\Omega}_{t+1}$ and $\overline{\Lambda}_{t+1}$ in equation \eqref{05.012}. By substituting equation \eqref{ad015} into approximation equation \eqref{03004}, we get that a distribution of a log sum random variable of the discounted stock price $\overline{P}_{t+1}$ and the discounted dividend payment $\overline{d}_{t+1}$ is given by 
\begin{equation}\label{05.045}
\ln\big(\overline{P}_{t+1}+\overline{d}_{t+1}\big)~|~\mathcal{G}_t\sim \mathcal{N}\bigg(\ln\big(D_t\big)-\frac{1}{2}\mathcal{D}[\Sigma_{uu}]+(I_n-G_t^{-1})\tilde{d}_t+G_t^{-1}(\tilde{P}_t+h_t),\Sigma_{uu}\bigg)
\end{equation}
under the martingale probability measure $\mathbb{\tilde{P}}^\circ$. From the above equation, one can easily prove that $\overline{P}_{t+1}+\overline{d}_{t+1}$ is the martingale, i.e., $\tilde{\mathbb{E}}^\circ\big[\overline{P}_{t+1}+\overline{d}_{t+1}\big|\mathcal{G}_{t}\big]=\overline{P}_t+\overline{d}_t$. Hence, it follows from the well--known covariance formula of the multivariate log--normal random vector that 
\begin{eqnarray}\label{05.046}
\overline{\Omega}_{t+1}(\mathcal{G}_t)&:=&\tilde{\mathbb{E}}^\circ\big[(\Delta\overline{P}_{t+1}+\Delta\overline{d}_{t+1})(\Delta\overline{P}_{t+1}+\Delta\overline{d}_{t+1})'\big|\mathcal{G}_{t}\big]\nonumber\\
&=&\exp\{\Sigma_{uu}\}\odot\big(\overline{P}_t\overline{P}_t'+\overline{P}_t\overline{d}_t'+\overline{d}_t\overline{P}_t'+\overline{d}_t\overline{d}_t'\big).
\end{eqnarray}

According to equation \eqref{ad021}, conditional on $\mathcal{F}_t$, a distribution of the log price--to--book ratio process at time $t$ is given by 
\begin{equation}\label{03085}
\tilde{m}_t~|~\mathcal{F}_t\sim \mathcal{N}\big(\tilde{\tilde{m}}_{t|t}^\circ,\Sigma^\circ(\tilde{m}_t|t)\big),
\end{equation}
under the martingale probability measure $\mathbb{\tilde{P}}^\circ$, where $\tilde{\tilde{m}}_{t|t}^\circ:=\mathsf{J}_m\tilde{\tilde{m}}_{t|t}^{*\circ}$ and $\Sigma^\circ(\tilde{m}_t|t):=\mathsf{J}_m\Sigma(\tilde{m}_t^*|t)\mathsf{J}_m'$ are a mean and covariance matrix of the log price--to--book ratio process $\tilde{m}_t$ given $\mathcal{F}_t$.
Since the discounted price at time $t$ is represented by $\overline{P}_t=D_t\exp\big\{\tilde{m}_t+\ln(\mathsf{B}_t)\big\}$, if take $X_1=X_2=\tilde{m}_t+\ln(\mathsf{B}_t)$, $\mu_1=\mu_2=\tilde{\tilde{m}}_{t|t}+\ln(\mathsf{B}_t)$, $\Sigma_{11}=\Sigma_{12}=\Sigma_{21}=\Sigma_{22}=\Sigma(\tilde{m}_t|t)$, $L=0$, and $\alpha_1=\alpha_2=i_n$ in Lemma \ref{lem02}, then by the tower property of a conditional expectation, we obtain that
\begin{eqnarray}\label{03086}
\overline{\Omega}_{t+1}&=&D_t^2\exp\{\Sigma_{uu}\}\odot\big\{\exp\big(\Sigma^\circ(\tilde{m}_t|t)\big)\odot \big(\tilde{\mu}_{t|t}^\circ(\mathcal{F}_t)\tilde{\mu}_{t|t}^\circ(\mathcal{F}_t)'\big)\\
&+&\tilde{\mu}_{t|t}^\circ(\mathcal{F}_t)d_t'+d_t\tilde{\mu}_{t|t}^\circ(\mathcal{F}_t)'+d_td_t'\big\},\nonumber
\end{eqnarray}
where the expectation equals
\begin{equation}\label{ad025}
\tilde{\mu}_{t|t}^\circ(\mathcal{F}_t):=\tilde{\mathbb{E}}^\circ[P_t|\mathcal{F}_t]=\mathsf{B_t}\odot\exp\bigg(\tilde{\tilde{m}}_{t|t}^\circ+\frac{1}{2}\mathcal{D}\big[\Sigma^\circ(\tilde{m}_t|t)\big]\bigg)
\end{equation}

On the other hand, because $\tilde{d}_{t+1}$ is measurable with respect to $\sigma$--field $\mathcal{F}_t$, we have that 
\begin{eqnarray}\label{03099}
\overline{\Lambda}_{t+1}:=\tilde{\mathbb{E}}\big[\overline{P}_{t+1}\overline{H}_T'\big|\mathcal{F}_t^x\big]-\tilde{\mathbb{E}}\big[\overline{P}_t\overline{H}_T'\big|\mathcal{F}_t^x\big]+\Delta\overline{d}_{t+1}\tilde{\mathbb{E}}\big[\overline{H}_T'\big|\mathcal{F}_t^x\big].
\end{eqnarray}
Since the $\sigma$--fields $\mathcal{F}_t$ and $\mathcal{T}_t^x$ are independent and $\overline{P}_{t+1}=D_{t+1}\exp\{\tilde{P}_{t+1}\}$, due to the $(t,u)$--forward probability measure $\hat{\mathbb{P}}_{t,u}^\circ$, the conditional covariance is
\begin{itemize}
\item[($i$)] for the Black--Scholes call and put options,
\begin{eqnarray}\label{03100}
\overline{\Lambda}_{t+1}&=& D_t^2B_{t,T}^\circ(\mathcal{F}_t)\Big\{\beta_{t+1}\hat{\mathbb{E}}_{t,T}^\circ\big[\exp(\tilde{P}_{t+1})H_T'\big|\mathcal{F}_t\big]\nonumber\\
&-&\hat{\mathbb{E}}_{t,T}^\circ\big[\exp(\tilde{P}_t)H_T'\big|\mathcal{F}_t\big]+\delta_{t+1}\hat{\mathbb{E}}_{t,T}^\circ\big[H_T'\big|\mathcal{F}_t\big]\Big\},
\end{eqnarray}
\item[($ii$)] for the equity--linked pure endowment insurances,
\begin{eqnarray}\label{03101}
\overline{\Lambda}_{t+1}&=& D_t^2B_{t,T}^\circ(\mathcal{F}_t)\Big\{\beta_{t+1}\hat{\mathbb{E}}_{t,T}^\circ\big[\exp(\tilde{P}_{t+1})f(P_T)'\big|\mathcal{F}_t\big]\nonumber\\
&-&\hat{\mathbb{E}}_{t,T}^\circ\big[\exp(\tilde{P}_t)f(P_T)'\big|\mathcal{F}_t\big]+\delta_{t+1}\hat{\mathbb{E}}_{t,T}^\circ\big[f(P_T)'\big|\mathcal{F}_t^x\big]\Big\}{}_{T-t}p_{x+t},
\end{eqnarray}
\item[($iii$)] and for the equity--linked term life insurances,
\begin{eqnarray}\label{03102}
\overline{\Lambda}_{t+1}&=& D_t^2\sum_{k=t}^{T-1}B_{t,k+1}^\circ(\mathcal{F}_t)\Big\{\beta_{t+1}\hat{\mathbb{E}}_{t,k+1}^\circ\big[\exp(\tilde{P}_{t+1})f(P_{k+1})'\big|\mathcal{F}_t\big]\nonumber\\
&-&\hat{\mathbb{E}}_{t,k+1}^\circ\big[\exp(\tilde{P}_t)f(P_{k+1})'\big|\mathcal{F}_t\big]+\delta_{t+1}\hat{\mathbb{E}}_{t,k+1}^\circ\big[f(P_{k+1})'\big|\mathcal{F}_t^x\big]\Big\}{}_{k-t}p_{x+t}q_{x+k},
\end{eqnarray}
where $\beta_{t+1}:=\exp\{-\tilde{r}_{t+1}\}$ and $\delta_{t+1}:=\beta_{t+1}d_{t+1}-d_t$ are $(n\times 1)$ processes, which are measurable with respect to $\sigma$--field $\mathcal{F}_t$.
\end{itemize}

In order to obtain the locally risk--minimizing strategies for the Black--Scholes call and put options, and the equity--linked life insurance products, we need to calculate the conditional expectations given in equations \eqref{03100}--\eqref{03102} for contingent claims $H_T=[P_T-K]^+$ and $H_T=[K-P_T]^+$, and sum insureds $f(P_k)=F_k\odot[P_k-L_k]^++G_k$ and $f(P_k)=F_k\odot[L_k-P_k]^+$ for $k=t+1,\dots,T$. Thus, we need Lemma \ref{lem02}, see Technical Annex. To use the Lemma, we need the following expectations and covariance matrices:
due to equation \eqref{ad006}, \eqref{03092}, \eqref{ad022}, and \eqref{ad024}, we have that for $k=t+1,\dots,T$,
\begin{equation}\label{ad026}
\hat{\mu}_{k|t,u}^{\tilde{P}\circ}(\mathcal{F}_t):=\mathbb{\hat{E}}_{t,u}^\circ\big[\tilde{P}_{k}\big|\mathcal{F}_t\big]=K_{k|t}\hat{\mu}_{t,u}^{c\circ}+\ln(\mathsf{B}_t),
\end{equation}
\begin{equation}\label{ad027}
\Sigma_{k|t}^{\tilde{P}\circ}(\mathcal{F}_t):=\widehat{\text{Var}}^\circ\big[\tilde{P}_{k}\big|\mathcal{F}_t\big]=K_{k|t}\Sigma_t^{c\circ}(\mathcal{F}_t)K_{k|t}',
\end{equation}
\begin{equation}\label{ad028}
\Sigma_{\tilde{P}_{t+1},\tilde{P}_k}^\circ(\mathcal{F}_t):=\widehat{\text{Cov}}^\circ\big[\tilde{P}_{t+1},\tilde{P}_k\big|\mathcal{F}_t\big]=K_{t+1|t}\Sigma_t^{c\circ}(\mathcal{F}_t)K_{k|t}',
\end{equation}
\begin{equation}\label{ad029}
\hat{\mu}_{t|t,u}^\circ(\mathcal{F}_t):=\hat{\mathbb{E}}_{t,u}^\circ\big[\tilde{P}_t\big|\mathcal{F}_t\big]=\mathsf{J}_m\hat{\tilde{m}}_{t|t,u}^{*\circ}+\ln(\mathsf{B}_t)
\end{equation}
\begin{equation}\label{ad030}
\Sigma_{t|t}^{\tilde{P}\circ}(\mathcal{F}_t):=\widehat{\text{Var}}^\circ\big[\tilde{P}_{t}\big|\mathcal{F}_t\big]=\mathsf{J}_m\Sigma^\circ(\tilde{m}_t^*|t)\mathsf{J}_m',
\end{equation}
and
\begin{equation}\label{ad031}
\Sigma_{\tilde{P}_t,\tilde{P}_k}^\circ(\mathcal{F}_t):=\widehat{\text{Cov}}^\circ\big[\tilde{m}_t+\ln(\mathsf{B}_t),\tilde{P}_k\big|\mathcal{F}_t\big]=\mathsf{J}_m\Sigma^\circ(\tilde{m}_t^*|t)(\Pi_{t,k}^{*m\circ})'\mathsf{J}_x'.
\end{equation}
Consequently, it follows from Lemma \ref{lem02} and equations \eqref{03100}--\eqref{ad031} that for $t=0,\dots,T-1$, $\overline{\Lambda}_{t+1}$s, which correspond to the call and put options and equity--linked life insurance products are obtained by the following equations
\begin{itemize}
\item[1.] for the dividend--paying Black--Scholes call option on the weighted asset price, we have
\begin{eqnarray}\label{03112}
\overline{\Lambda}_{t+1}&=& D_t^2B_{t,T}^\circ(\mathcal{F}_t)\bigg\{\beta_{t+1}\Psi^+\Big(K;i_n;i_n;\hat{\mu}_{t+1|t,T}^{\tilde{P}\circ};\hat{\mu}_{T|t,T}^{\tilde{P}\circ};\Sigma_{t+1|t}^{\tilde{P}\circ};\Sigma_{\tilde{P}_{t+1},\tilde{P}_T}^\circ;\Sigma_{T|t}^{\tilde{P}\circ}\Big)\\
&-&\Psi^+\Big(K;i_n;i_n;\hat{\mu}_{t|t,T}^{\tilde{P}\circ};\hat{\mu}_{T|t,T}^{\tilde{P}\circ};\Sigma_{t|t}^{\tilde{P}\circ};\Sigma_{\tilde{P}_{t},\tilde{P}_T}^\circ;\Sigma_{T|t}^{\tilde{P}\circ}\Big)+\Psi^+\Big(K;\delta_{t+1};i_n;0;\hat{\mu}_{T|t,T}^{\tilde{P}\circ};0;0;\Sigma_{T|t}^{\tilde{P}\circ}\Big)\bigg\},\nonumber
\end{eqnarray}
\item[2.] for the dividend--paying Black--Scholes put option on the weighted asset price, we have
\begin{eqnarray}\label{03113}
\overline{\Lambda}_{t+1}&=& D_t^2B_{t,T}^\circ(\mathcal{F}_t)\bigg\{\beta_{t+1}\Psi^-\Big(K;i_n;i_n;\hat{\mu}_{t+1|t,T}^{\tilde{P}\circ};\hat{\mu}_{T|t,T}^{\tilde{P}\circ};\Sigma_{t+1|t}^{\tilde{P}\circ};\Sigma_{\tilde{P}_{t+1},\tilde{P}_T}^\circ;\Sigma_{T|t}^{\tilde{P}\circ}\Big)\\
&-&\Psi^-\Big(K;i_n;i_n;\hat{\mu}_{t|t,T}^{\tilde{P}\circ};\hat{\mu}_{T|t,T}^{\tilde{P}\circ};\Sigma_{t|t}^{\tilde{P}\circ};\Sigma_{\tilde{P}_{t},\tilde{P}_T}^\circ;\Sigma_{T|t}^{\tilde{P}\circ}\Big)+\Psi^-\Big(K;\delta_{t+1};i_n;0;\hat{\mu}_{T|t,T}^{\tilde{P}\circ};0;0;\Sigma_{T|t}^{\tilde{P}\circ}\Big)\bigg\},\nonumber
\end{eqnarray}
\item[3.] for the $T$-year guaranteed term life insurance, corresponding to a segregated fund contract, we have
\begin{eqnarray}\label{03114}
\overline{\Lambda}_{t+1}&=& D_t^2\sum_{k=t}^{T-1}B_{t,k+1}^\circ(\mathcal{F}_t)\bigg\{\beta_{t+1}\Psi^-\Big(L_{k+1};i_n;F_{k+1};\hat{\mu}_{t+1|t,k+1}^{\tilde{P}\circ};\hat{\mu}_{k+1|t,k+1}^{\tilde{P}\circ};\Sigma_{t+1|t}^{\tilde{P}\circ};\Sigma_{\tilde{P}_{t+1},\tilde{P}_{k+1}}^\circ;\nonumber\\
&&\Sigma_{k+1|t}^{\tilde{P}\circ}\Big)-\Psi^-\Big(L_{k+1};i_n;F_{k+1};\hat{\mu}_{t|t,k+1}^{\tilde{P}\circ};\hat{\mu}_{k+1|t,k+1}^{\tilde{P}\circ};\Sigma_{t|t}^{\tilde{P}\circ};\Sigma_{\tilde{P}_{t},\tilde{P}_{k+1}}^\circ;\Sigma_{k+1|t}^{\tilde{P}\circ}\Big)\\
&+&\Psi^-\Big(L_{k+1};\delta_{t+1};F_{k+1};0;\hat{\mu}_{k+1|t,k+1}^{\tilde{P}\circ};0;0;\Sigma_{k+1|t}^{\tilde{P}\circ}\Big)\bigg\}{}_{k-t}p_{x+t}q_{x+k},\nonumber
\end{eqnarray}
\item[4.] for the $T$-year guaranteed pure endowment insurance, corresponding to a segregated fund contract, we have
\begin{eqnarray}\label{03115}
\overline{\Lambda}_{t+1}&=& D_t^2B_{t,T}^\circ(\mathcal{F}_t)\bigg\{\beta_{t+1}\Psi^-\Big(L_T;i_n;F_T;\hat{\mu}_{t+1|t,T}^{\tilde{P}\circ};\hat{\mu}_{T|t,T}^{\tilde{P}\circ};\Sigma_{t+1|t}^{\tilde{P}\circ};\Sigma_{\tilde{P}_{t+1},\tilde{P}_T}^\circ;\Sigma_{T|t}^{\tilde{P}\circ}\Big)\nonumber\\
&-&\Psi^-\Big(L_T;i_n;F_T;\hat{\mu}_{t|t,T}^{\tilde{P}\circ};\hat{\mu}_{T|t,T}^{\tilde{P}\circ};\Sigma_{t|t}^{\tilde{P}\circ};\Sigma_{\tilde{P}_{t},\tilde{P}_{T}}^\circ;\Sigma_{T|t}^{\tilde{P}\circ}\Big)\\
&+&\Psi^-\Big(L_T;\delta_{t+1};F_T;0;\hat{\mu}_{T|t,T}^{\tilde{P}\circ};0;0;\Sigma_{T|t}^{\tilde{P}\circ}\Big)\bigg\}{}_{T-t}p_{x+t},\nonumber
\end{eqnarray}
\item[5.] for the $T$-year guaranteed unit--linked term life insurance, we have
\begin{eqnarray}\label{03116}
\overline{\Lambda}_{t+1}&=& D_t^2\sum_{k=t}^{T-1}B_{t,k+1}^\circ(\mathcal{F}_t)\bigg\{\beta_{t+1}\Psi^+\Big(L_{k+1};i_n;F_{k+1};\hat{\mu}_{t+1|t,k+1}^{\tilde{P}\circ};\hat{\mu}_{k+1|t,k+1}^{\tilde{P}\circ};\Sigma_{t+1|t}^{\tilde{P}\circ};\Sigma_{\tilde{P}_{t+1},\tilde{P}_{k+1}}^\circ;\nonumber\\
&&\Sigma_{k+1|t}^{\tilde{P}\circ}\Big)-\Psi^+\Big(L_{k+1};i_n;F_{k+1};\hat{\mu}_{t|t,k+1}^{\tilde{P}\circ};\hat{\mu}_{k+1|t,k+1}^{\tilde{P}\circ};\Sigma_{t|t}^{\tilde{P}\circ};\Sigma_{\tilde{P}_{t},\tilde{P}_{k+1}}^\circ;\Sigma_{k+1|t}^{\tilde{P}\circ}\Big)\nonumber\\
&+&\Psi^+\Big(L_{k+1};\delta_{t+1};F_{k+1};0;\hat{\mu}_{k+1|t,k+1}^{\tilde{P}\circ};0;0;\Sigma_{k+1|t}^{\tilde{P}\circ}\Big)+\bigg[\beta_{t+1}\exp\bigg\{\hat{\mu}_{t+1|t,k+1}^{\tilde{P}\circ}+\frac{1}{2}\mathcal{D}\big[\Sigma_{t+1|t}^{\tilde{P}\circ}\big]\bigg\}\nonumber\\
&-&\exp\bigg\{\hat{\mu}_{t|t,k+1}^{\tilde{P}\circ}+\frac{1}{2}\mathcal{D}\big[\Sigma_{t|t}^{\tilde{P}\circ}\big]\bigg\}+\delta_{t+1}\bigg]G_{k+1}'\bigg\}{}_{k-t}p_{x+t}q_{x+k},
\end{eqnarray}
\item[6.] and for the $T$-year guaranteed unit--linked pure endowment insurance, we have
\begin{eqnarray}\label{03117}
\overline{\Lambda}_{t+1}&=& D_t^2B_{t,T}^\circ(\mathcal{F}_t)\bigg\{\beta_{t+1}\Psi^+\Big(L_T;i_n;F_T;\hat{\mu}_{t+1|t,T}^{\tilde{P}\circ};\hat{\mu}_{T|t,T}^{\tilde{P}\circ};\Sigma_{t+1|t}^{\tilde{P}\circ};\Sigma_{\tilde{P}_{t+1},\tilde{P}_T}^\circ;\Sigma_{T|t}^{\tilde{P}\circ}\Big)\nonumber\\
&-&\Psi^+\Big(L_T;i_n;F_T;\hat{\mu}_{t|t,T}^{\tilde{P}\circ};\hat{\mu}_{T|t,T}^{\tilde{P}\circ};\Sigma_{t|t}^{\tilde{P}\circ};\Sigma_{\tilde{P}_{t},\tilde{P}_T}^\circ;\Sigma_{T|t}^{\tilde{P}\circ}\Big)\nonumber\\
&+&\Psi^+\Big(L_T;\delta_{t+1};F_T;0;\hat{\mu}_{T|t,T}^{\tilde{P}\circ};0;0;\Sigma_{T|t}^{\tilde{P}\circ}\Big)+\bigg[\beta_{t+1}\exp\bigg\{\hat{\mu}_{t+1|t,T}^{\tilde{P}\circ}+\frac{1}{2}\mathcal{D}\big[\Sigma_{t+1|t}^{\tilde{P}\circ}\big]\bigg\}\nonumber\\
&-&\exp\bigg\{\hat{\mu}_{t|t,T}^{\tilde{P}\circ}+\frac{1}{2}\mathcal{D}\big[\Sigma_{t|t}^{\tilde{P}\circ}\big]\bigg\}+\delta_{t+1}\bigg]G_T'\bigg\}{}_{T-t}p_{x+t},
\end{eqnarray}
\end{itemize}
where the functions $\Psi^+$ and $\Psi^-$ are defined in Lemma \ref{lem02}. Note that to simplify the notations in equations \eqref{03112}--\eqref{03117}, for the equations, we omit the notation $(\mathcal{F}_t)$ in equations \eqref{ad026}--\eqref{ad031}. As a result, by substituting equations \eqref{03086} and \eqref{03112}--\eqref{03117} into equation \eqref{05.012} we can obtain the locally risk--minimizing strategies for the Black--Scholes call and put options and the equity--linked life insurance products corresponding to the private company.

\section{Parameter Estimation}

The Kalman filtering, which was introduced by \citeA{Kalman60} is an algorithm that provides estimates of parameters of observed and unobserved (state) processes. The Kalman filtering has been demonstrating its usefulness in various applications. It has been used extensively in economics, system theory, the physical sciences, and engineering. It is usually assume that: 
\begin{itemize}
\item[(i)] a vector of observed variables is represented in terms of the vector of state variables in linear form (measurement equation)
\item[(ii)] and a vector of state variables is governed by the VAR(1) process (transition equation),
\end{itemize}
see \citeA{Hamilton94} and \citeA{Lutkepohl05}.

Let us reconsider the log private company valuation model \eqref{03017}. In state--space form, the system can be written as
\begin{equation}\label{03118}
\begin{cases}
y_t=\nu_{y,t}+\Psi_{y,t}\tilde{m}_t^*+A_{y}z_{t-1}^*+G_{y,t}\eta_t\\
\tilde{m}_t^*=\nu_{m,t}^*+C\tilde{m}_{t-1}^*+\tilde{w}_t^*
\end{cases}~~~\text{for}~t=1,\dots,T
\end{equation}
under the real probability measure $\mathbb{P}$, where $y_t=(\tilde{b}_t',z_t')'$ is the $([n+\ell]\times 1)$ process of observed variables, $\nu_{y,t}:=\big(\nu_{b,t}',\nu_{z,t}'\big)'$ is an ($[n+\ell]\times 1$) intercept process of the observed variables process $y_t$, $\eta_t=(u_t',v_t')'$ is the ($[n+\ell]\times 1$) white noise process of the process $y_t$, $\tilde{m}_t^*=(\tilde{m}_t,\tilde{m}_{t-1})'$ is the ($2n\times 1$) state process that consists of the log price--to--book ratio process at times $t$ and $t-1$, $\nu_{m,t}^*=(\nu_{m,t}',0)'$ is the ($2n\times 1$) intercept process of the state process $\tilde{m}_t^*$, $w_t^*=(w_t',0)'$ is the ($2n\times 1$) white noise process of the state process $\tilde{m}_t^*$, and
\begin{equation}\label{03119}
\Psi_{y,t}:=\begin{bmatrix}
\Psi_{b,t}\\
0
\end{bmatrix},~~~A_{y}:=\begin{bmatrix}
0\\
A
\end{bmatrix},~~~\Sigma_{\eta\eta}=\begin{bmatrix}
\Sigma_{uu} & \Sigma_{uv}\\
\Sigma_{vu} & \Sigma_{vv}
\end{bmatrix},
\end{equation}
\begin{equation}\label{03120}
G_{y,t}=\begin{bmatrix}
G_t & 0\\
0 & I_\ell
\end{bmatrix},~~~C=\begin{bmatrix}
I_n & 0\\
I_n & 0
\end{bmatrix},~~~\text{and}~~~\Sigma_{w^*w^*}:=\begin{bmatrix}
\Sigma_{ww} & 0\\
0 & 0
\end{bmatrix}
\end{equation}
are ($[n+\ell]\times 2n$), $([n+\ell]\times \ell p)$, $([n+\ell]\times [n+\ell])$, $([n+\ell]\times [n+\ell])$, $(2n\times 2n)$, and $(2n\times 2n)$ matrices, respectively. For system \eqref{03118}, its first line determines the measurement equation and the second line determines the transition equation. 

\subsection{Review of the Kalman Filtering}

For each $t=0,\dots,T$, conditional on the information $\mathcal{F}_t$, conditional expectations and covariance matrices of the observed variables vectors and the state vectors are recursively obtained by the Kalman filtering (see \citeA{Hamilton94} and \citeA{Lutkepohl05}):
\begin{itemize}
\item Initialization: 
\begin{itemize}
\item Expectation
\begin{equation}\label{03121}
\tilde{m}_{0|0}^*:=\mathbb{E}[\tilde{m}_0^*|\mathcal{F}_0]=(\mu_0',\mu_0')'
\end{equation}
\item Covariance
\begin{equation}\label{03122}
\Sigma(\tilde{m}_0^*|0):=\text{Cov}[\tilde{m}_0^*|\mathcal{F}_0]=(\mathsf{E}_2\otimes\Sigma_0)
\end{equation}
\end{itemize}
\item Prediction step: for $t=1,\dots,T$, 
\begin{itemize}
\item Expectations
\begin{eqnarray}
\tilde{m}_{t|t-1}^*&:=&\mathbb{E}(\tilde{m}_t^*|\mathcal{F}_{t-1})=\nu_{m,t}^*+C\tilde{m}_{t-1|t-1}^*\label{03123}\\
y_{t|t-1}&:=&\mathbb{E}(y_t|\mathcal{F}_{t-1})=\nu_{y,t}+\Psi_{y,t}\tilde{m}_{t|t-1}^*+A_yz_{t-1}^*\label{03124}
\end{eqnarray}
\item Covariances
\begin{eqnarray}
\Sigma(\tilde{m}_t^*|t-1)&:=&\text{Cov}[\tilde{m}_t^*|\mathcal{F}_{t-1}]=C\Sigma(\tilde{m}_{t-1}^*|t-1)C'+\Sigma_{w^*w^*}\label{03125}\\
\Sigma(y_t|t-1)&:=&\text{Var}[y_{t}|\mathcal{F}_{t-1}]=\Psi_{y,t}\Sigma(\tilde{m}_{t}|t-1)\Psi_{y,t}'+G_{y,t}\Sigma_{\eta\eta}G_{y,t}'\label{03126}
\end{eqnarray}
\end{itemize}
\item Correction step: for $t=1,\dots,T$, 
\begin{itemize}
\item Expectations
\begin{equation}\label{03127}
\tilde{m}_{t|t}^*:=\mathbb{E}[\tilde{m}_t^*|\mathcal{F}_t]=\tilde{m}_{t|t-1}^*+\mathcal{K}_{t}\big(y_t-y_{t|t-1}\big)
\end{equation}
\item Covariances
\begin{equation}\label{03128}
\Sigma(\tilde{m}_t^*|t):=\text{Cov}[\tilde{m}_t^*|\mathcal{F}_t]=\Sigma(\tilde{m}_t^*|t-1)-\mathcal{K}_{t}\Sigma(y_t|t-1)\mathcal{K}_t',
\end{equation}
where $\mathcal{K}_t:=\Sigma(\tilde{m}_t^*|t-1)\Psi_{y,t}'\Sigma(y_t|t-1)^{-1}$ is the Kalman filter gain.
\end{itemize}
\end{itemize}

For each $t=T+1,T+2,\dots$, conditional on the information $\mathcal{F}_T$, conditional expectations and covariance matrices of the observed variables vectors and the state vectors are recursively obtained by (see \citeA{Hamilton94} and \citeA{Lutkepohl05}):

\begin{itemize}
\item Forecasting step: for $t=T+1,T+2,\dots$, 
\begin{itemize}
\item Expectations
\begin{eqnarray}
\tilde{m}_{t|T}^*&:=&\mathbb{E}[\tilde{m}_t^*|\mathcal{F}_T]=\nu_{m,t}^*+C\tilde{m}_{t-1|T}^*\label{03129}\\
y_{t|T}&:=&\mathbb{E}[y_t|\mathcal{F}_T]=\nu_{y,t}+\Psi_{y,t}\tilde{m}_{t|T}^*+A_{y}z_{t-1}^*\label{03130}
\end{eqnarray}
\item Covariances
\begin{eqnarray}
\Sigma(\tilde{m}_t^*|T)&:=&\text{Cov}[\tilde{m}_t^*|\mathcal{F}_T]=C\Sigma(\tilde{m}_{t-1}^*|T)C'+\Sigma_{w^*w^*}\label{03131}\\
\Sigma(y_t|T)&:=&\text{Var}[y_t|\mathcal{F}_T]=\Psi_{y,t}\Sigma(\tilde{m}_t^*|T)\Psi_{y,t}'+G_{y,t}\Sigma_{\eta\eta}G_{y,t}'\label{03132}
\end{eqnarray}
\end{itemize}
\end{itemize}

The Kalman filtering, which is considered above provides an algorithm for filtering the state vector of price--to--book ratio $\tilde{m}_t^*$. To estimate the parameters of our model \eqref{03118}, in addition to the Kalman filtering, we also need to make inferences about the state vector of price--to--book ratio $\tilde{m}_t^*$ for each $t=1,\dots,T$ based on the full information $\mathcal{F}_T$, see below. Such an inference is called the smoothed estimate of the state vector of price--to--book ratio $\tilde{m}_t^*$. The smoothed inference of the state vector can be obtained by the following Kalman smoother recursions, see \citeA{Hamilton94} and \citeA{Lutkepohl05}.

\begin{itemize}
\item Smoothing step: for $t=T-1,T-2,\dots,0$,
\begin{itemize}
\item Expectations
\begin{equation}\label{03133}
\tilde{m}_{t|T}^*:=\mathbb{E}[\tilde{m}_t^*|\mathcal{F}_T]=\tilde{m}_{t|t}^*+\mathcal{S}_{t}\big(\tilde{m}_{t+1|T}^*-\tilde{m}_{t+1|t}^*\big)
\end{equation}
\item Covariances
\begin{equation}\label{03134}
\Sigma(\tilde{m}_t^*|T):=\text{Cov}[\tilde{m}_t^*|\mathcal{F}_T]=\Sigma(\tilde{m}_t^*|t)-\mathcal{S}_{t}\big(\Sigma(\tilde{m}_{t+1}^*|t)-\Sigma(\tilde{m}_{t+1}^*|T)\big)\mathcal{S}_t',
\end{equation}
where $\mathcal{S}_{t}:=\Sigma(\tilde{m}_t^*|t)C'\Sigma^{-1}(\tilde{m}_{t+1}^*|t)$ is the Kalman smoother gain. 
\end{itemize}
\end{itemize}

\subsection{EM Algorithm}

In the Expectation Maximization (EM) algorithm, one considers a joint density function of a random vector, which is composed of observed variables and state variables. In our cases, the vectors of observed variables and the state variables correspond to a vector of the observed variables that consist of the log book value growth rates and economic variables, $y:=(y_1',\dots,y_T')'$, and a vector of state variables that consist of the log price--to--book ratios, $\tilde{m}:=(\tilde{m}_0',\dots,\tilde{m}_T')'$, respectively. Usages of the EM algorithm in econometrics can be found in \citeA{Hamilton90} and \citeA{Schneider92}. Let us denote the joint density function by $f_{y,\tilde{m}}(y,\tilde{m})$. The EM algorithm consists of two steps. 

In the expectation (E) step of the EM algorithm, one has to determine a form of an expectation of log of the joint density given the full information $\mathcal{F}_T$ and maximum likelihood (ML) estimation at iteration $j$ of a parameter vector. We denote the expectation by $\Lambda(\theta|\mathcal{F}_T;\theta^{(j)})$, that is, $\Lambda(\theta|\mathcal{F}_T;\theta^{(j)}):=\mathbb{E}\big[\ln\big(f_{y,\tilde{m}}(y,\tilde{m})\big)|\mathcal{F}_T;\theta^{(j)}\big]$. For our log private company valuation model \eqref{03017}, an absolute value of a determinant formed $\big([\tilde{n}T]\times[\tilde{n}T]\big)$ matrix of partial derivatives of the elements of a vector $(u',v',w')'$ with respect to elements of a vector $(\tilde{b}',z',\tilde{m}')'$ equals $\Big|\frac{\partial (u',v',w')}{\partial (\tilde{b}',z',\tilde{m}')}\Big|=\prod_{t=1}^T|G_t^{-1}|$. Therefore, conditional on the full information $\mathcal{F}_T$ and ML estimator at iteration $j$, the expectation of log of the joint density of the vectors of the observed variables and the log price--to--book ratios is given by the following equation
\begin{eqnarray}\label{03136}
\Lambda(\theta|\mathcal{F}_T;\theta^{(j)})&=&-\frac{\tilde{n}T+n}{2}\ln(2\pi)-\frac{T}{2}\ln(|\Sigma_{\eta\eta}|)-\frac{T}{2}\ln(|\Sigma_{ww}|)-\frac{1}{2}\ln(|\Sigma_{0}|)\nonumber\\
&-&\frac{1}{2}\sum_{t=1}^T\mathbb{E}\Big[u_t'\Omega_{uu}u_t\Big|\mathcal{F}_T;\theta^{(j)}\Big]-\sum_{t=1}^T\mathbb{E}\Big[u_t'\Omega_{uv}v_t\Big|\mathcal{F}_T;\theta^{(j)}\Big]\nonumber\\
&-&\frac{1}{2}\sum_{t=1}^T\mathbb{E}\Big[v_t'\Omega_{vv} v_t\Big|\mathcal{F}_T;\theta^{(j)}\Big]-\frac{1}{2}\sum_{t=1}^T\mathbb{E}\Big[w_t'\Sigma_{ww}^{-1} w_t\Big|\mathcal{F}_T;\theta^{(j)}\Big]\\
&-&\frac{1}{2}\mathbb{E}\Big[\big(\tilde{m}_0-\mu_0\big)'\Sigma_0^{-1}\big(\tilde{m}_0-\mu_0\big)\Big|\mathcal{F}_T;\theta^{(j)}\Big]-\sum_{t=1}^T\ln(|G_t|),\nonumber
\end{eqnarray}
where $\theta:=\big(\text{vec}(C_k)',\mu_0',\text{vec}(C_z)',\text{vec}(A)',\text{vec}(C_m)',\text{vec}(\Sigma_{\eta\eta})',\text{vec}(\Sigma_{ww})',\text{vec}(\Sigma_0)'\big)'$ is a vector, which consists of all parameters of the model \eqref{03017}, $\Omega_{uu}$, $\Omega_{uv}$, and $\Omega_{vv}$ are the partitions of the matrix $\Sigma_{\eta\eta}^{-1}$, corresponding to the white noise process $\eta_t=(u_t',v_t')'$, $u_t=G_t^{-1}(\tilde{b}_t+\tilde{m}_t-\tilde{\Delta}_t)-\tilde{m}_{t-1}+\tilde{\Delta}_t-C_k\psi_t+G_t^{-1}h_t$ is the $(n\times 1)$ white noise process of the log book value growth rate process $\tilde{b}_t$, $v_t=z_t-C_z\psi_t-Az_{t-1}^*$ is the $(\ell\times 1)$ white noise process of the economic variables process $z_t$, and $w_t=\tilde{m}_t-C_m\psi_t-\tilde{m}_{t-1}$ is the $(n\times 1)$ white noise process of the log price--to--book value process $\tilde{m}_t$. 

In the maximization (M) step of the EM algorithm, we need to find a maximum likelihood estimation $\theta^{(j+1)}$ that maximizes the expectation, which is determined in the E step. First, we consider the case, where the dividend proportional to the book value of equity, see equation \eqref{03010}. Let us define a vector and matrices, which deal with partial derivatives of the log--likelihood (objective) function $\Lambda(\theta|\mathcal{F}_T;\theta^{(j)})$ with respect to parameters $\mu_0$, $C_k$ and $C_m$:
\begin{equation}\label{03137}
\delta_t:=(G_t-I_n)\Bigg(u_t+\tilde{m}_{t-1}-\bigg(\mu_0+C_m\sum_{s=1}^{t-1}\psi_s\bigg)\Bigg)
\end{equation}
is an $(n\times 1)$ vector and
\begin{equation}\label{03138}
L_i:=\text{diag}\{0,\dots,0,1,0,\dots,0\},~~~i=1,\dots,n
\end{equation}
are $(n\times n)$ diagonal matrix, whose $i$--th diagonal element is one and others are zero, where we use a convention $\sum_{s=1}^{0}\psi_s=0$. A conditional expectation of the vector $\delta_t$ given full information $\mathcal{F}_T$ is
\begin{equation}\label{03139}
\delta_{t|T}:=\mathbb{E}\big[\delta_t\big|\mathcal{F}_T\big]=(G_t-I_n)\Bigg(u_{t|T}+\tilde{m}_{t-1|T}-\bigg(\mu_0+C_m\sum_{s=1}^{t-1}\psi_s\bigg)\Bigg),
\end{equation}
where $\tilde{m}_{t|T}:=\mathsf{J}_m\tilde{m}_{t|T}^{*}$ is an $(n\times 1)$ smoothed inference of the log price--to-book value ratio process $\tilde{m}_t$ and $u_{t|T}=G_t^{-1}(\tilde{b}_t+\tilde{m}_{t|T}-\tilde{\Delta}_t)-\tilde{m}_{t-1|T}+\tilde{\Delta}_t-C_k\psi_t+G_t^{-1}h_t$ is an $(n\times 1)$ smoothed inference of the white noise process $u_t$. Since $\delta_t-\delta_{t|T}=(G_t-I_n)(\mathsf{R}_t+\mathsf{J}_m^c)\big(\tilde{m}_t^*-\tilde{m}_{t|T}^*\big)$ and $u_t-u_{t|T}=\mathsf{R}_t\big(\tilde{m}_t^*-\tilde{m}_{t|T}^*\big)$, it is clear that
\begin{equation}\label{03143}
\mathbb{E}\big[\delta_tu_t'|\mathcal{F}_T\big]=Z_t+\delta_{t|T}u_{t|T}'
\end{equation}
where $\mathsf{J}_m^c:=[0:I_n]$ is an $(n\times 2n)$ matrix, which is used to extract $\tilde{m}_{t-1}$ from $\tilde{m}_t^*$, $\mathsf{R}_t:=[G_t^{-1}:-I_n]$ is an $(n\times 2n)$ matrix and $Z_t:=(G_t-I_n)(\mathsf{R}_t+\mathsf{J}_m^c)\Sigma(\tilde{m}_t^*|T)\mathsf{R}_t'$ is an $(n\times n)$ matrix. 

Let $c_k:=\text{vec}(C_k)$ and $c_m:=\text{vec}(C_m)$ be vectorizations of the parameter matrices $C_k$ and $C_m$. Then, according to Lemma \ref{lem03}, partial derivatives of $\Lambda(\theta|\mathcal{F}_T;\theta^{(j)})$ with respect to the parameters $\mu_0$, $c_k$ and $c_m$ are obtained by
\begin{eqnarray}\label{03140}
\frac{\partial \Lambda(\theta|\mathcal{F}_T;\theta^{(j)})}{\partial \mu_0}&=&\sum_{t=1}^T\big(g_t-i_n\big)+\Sigma_0^{-1}\big(\tilde{m}_{0|T}^{(j)}-\mu_0\big)-\sum_{t=1}^T\text{diag}\big\{\delta_{t|T}^{(j)}\big\}\Omega_{uv}v_t\nonumber\\
&-&\sum_{t=1}^T\mathbb{E}\Big[\big[\text{tr}\{\delta_tu_t'L_{uu}^1\}:\dots: \text{tr}\{\delta_tu_t'L_{uu}^n\}\big]\Big|\mathcal{F}_T;\theta^{(j)}\Big]',
\end{eqnarray}
\begin{eqnarray}\label{03141}
\frac{\partial \Lambda(\theta|\mathcal{F}_T;\theta^{(j)})}{\partial c_k}&=&\sum_{t=1}^T(\psi_t\otimes I_n)\big(g_t-i_n\big)-\sum_{t=1}^T(\psi_t\otimes I_n)\Big(\big(\text{diag}\big\{\delta_{t|T}^{(j)}\big\}-I_n\big)\Omega_{uv}v_t\Big)\\
&-&\sum_{t=1}^T(\psi_t\otimes I_n)\bigg(\mathbb{E}\Big[\big[\text{tr}\{\delta_tu_t'L_{uu}^1\}:\dots:\text{tr}\{\delta_tu_t'L_{uu}^n\}\big]\Big|\mathcal{F}_T;\theta^{(j)}\Big]'-\Omega_{uu}u_{t|T}^{(j)}\bigg),\nonumber
\end{eqnarray}
and
\begin{eqnarray}\label{03142}
\frac{\partial \Lambda(\theta|\mathcal{F}_T;\theta^{(j)})}{\partial c_m}&=&\sum_{t=1}^T\bigg(\sum_{s=1}^{t-1}(\psi_s\otimes I_n)\bigg)\big(g_t-i_n\big)+\sum_{t=1}^T (\psi_t\otimes I_n)\Sigma_{ww}^{-1}w_{t|T}^{(j)}\nonumber\\
&-&\sum_{t=1}^T\bigg(\sum_{s=1}^{t-1}(\psi_s\otimes I_n)\bigg)\mathbb{E}\Big[\big[\text{tr}\{\delta_tu_t'L_{uu}^1\}:\dots:\text{tr}\{\delta_tu_t'L_{uu}^n\}\big]\Big|\mathcal{F}_T;\theta^{(j)}\Big]'\\
&-&\sum_{t=1}^T\bigg(\sum_{s=1}^{t-1}(\psi_s\otimes I_n)\bigg)\text{diag}\big\{\delta_{t|T}^{(j)}\big\}\Omega_{uv}v_t\nonumber,
\end{eqnarray}
where $u_{t|T}^{(j)}=G_t^{-1}\big(\tilde{b}_t+\tilde{m}_{t|T}^{(j)}-\tilde{\Delta}_t\big)-\tilde{m}_{t-1|T}^{(j)}+\tilde{\Delta}_t-C_k\psi_t+G_t^{-1}h_t$ is a smoothed white noise process at iteration $j$, corresponding to the white noise process $u_t$, 
\begin{equation}\label{ad032}
\delta_{t|T}^{(j)}:=\mathbb{E}\big[\delta_t\big|\mathcal{F}_T;\theta^{(j)}\big]=\big(G_t-I_n\big)\Bigg(u_{t|T}^{(j)}+\tilde{m}_{t-1|T}^{(j)}-\bigg(\mu_0+C_m\sum_{s=1}^{t-1}\psi_s\bigg)\Bigg)
\end{equation}
is a smoothed vector at iteration $j$, corresponding to the random vector $\delta_t$,  for $i=1,\dots,n$, $L_{uu}^i:=\Omega_{uu}L_i$ and $L_{vu}^i:=\Omega_{vu}L_i$ are $(n\times n)$ and $(\ell\times n)$ matrices, whose $i$--th column equal $i$--th column of the matrices $\Omega_{uu}$ and $\Omega_{vu}$ and other columns equal zero, respectively, $w_{t|T}^{(j)}:=\mathsf{R}\tilde{m}_{t|T}^{*(j)}-C_m\psi_t$ is a smoothed white noise process, corresponding to the white noise process $w_t$, $\mathsf{R}:=[I_n:-I_n]$ is an $(n\times 2n)$ matrix, and $\text{tr}\{O\}$ denotes the trace of a generic square matrix $O$. Note that it follows from Technical Annex that representations of the linearization parameter vectors $g_t$ and $h_t$ are given by 
\begin{equation}\label{ad033}
g_t=\exp\{-\varphi_t\}\oslash\big(\exp\{-\varphi_t\}-i_n\big)
\end{equation}
and
\begin{equation}\label{ad034}
h_t=-\Big\{\varphi_t\odot\exp\{-\varphi_t\}\oslash\Big(\exp\{-\varphi_t\}-i_n\Big)+\ln\Big(\exp\{-\varphi_t\}-i_n\Big)\Big\},
\end{equation}
respectively, where 
\begin{equation}\label{ad035}
\varphi_t:=\tilde{\Delta}_t-C_k\psi_t-\bigg(\mu_0+C_m\sum_{s=1}^{t-1}\psi_s\bigg).
\end{equation}

Let $u_{t|T}^k=G_t^{-1}(\tilde{b}_t+\tilde{m}_{t|T}-\tilde{\Delta}_t)-\tilde{m}_{t-1|T}+\tilde{\Delta}_t+G_t^{-1}h_t$ be a smoothed process, which excludes the term $C_k\psi_t$ from the smoothed white noise process $u_{t|T}$. Consequently, according to equation \eqref{03143}, the expectation, which is given in equations \eqref{03140}--\eqref{03142} equals
\begin{eqnarray}\label{03144}
\mathsf{E}_t^{(j)}&:=&\mathbb{E}\Big[\big[\text{tr}\{\delta_tu_t'L_{uu}^1\}:\dots:\text{tr}\{\delta_tu_t'L_{uu}^n\}\big]\Big|\mathcal{F}_T;\theta^{(j)}\Big]\nonumber\\
&=&\Big[\text{tr}\Big\{\Big(Z_t^{(j)}+\delta_{t|T}^{(j)}\big(u_{t|T}^{(j)}\big)'\Big)L_{uu}^1\Big\},\dots,\text{tr}\Big\{\Big(Z_t^{(j)}+\delta_{t|T}^{(j)}\big(u_{t|T}^{(j)}\big)'\Big)L_{uu}^n\Big\}\Big],
\end{eqnarray}
where $Z_t^{(j)}:=(G_t-I_n)\big(\mathsf{R}_t+\mathsf{J}_m^c)\Sigma^{(j)}(\tilde{m}_t^*|T)\mathsf{R}_t'$. By setting equations \eqref{03140}--\eqref{03142} to zero and taking into account the fact that for suitable matrices $A$, $B$, and $C$, the relationship $\text{vec}(ABC)=(C'\otimes A)\text{vec}(B)$ holds, the relationship $u_{t|T}^{(j)}=u_{t|T}^{k(j)}-C_k\psi_t$, and equation \eqref{03144}, one obtains ML equations of the parameters $\mu_0$, $C_k$, and $C_m$.
\begin{equation}\label{03145}
\mu_0:=\tilde{m}_{0|T}^{(j)}+\Sigma_0\sum_{t=1}^T\alpha_t^{(j)},
\end{equation}
\begin{eqnarray}\label{03146}
C_k:=\Omega_{uu}^{-1}\Bigg[\sum_{t=1}^T\bigg(\Omega_{uu}u_{t|T}^{k(j)}-\alpha_t^{(j)}-\Omega_{uv}v_t\bigg)\psi_t'\Bigg]\Bigg[\sum_{t=1}^T\psi_t\psi_t'\Bigg]^{-1},
\end{eqnarray}
and
\begin{eqnarray}\label{03147}
C_m:=\Bigg[\sum_{t=1}^T\bigg\{\Sigma_{ww}\alpha_t^{(j)}\sum_{s=1}^{t-1}\psi_s'+\mathsf{R}\tilde{m}_{t|T}^{*(j)}\psi_t'\bigg\}\Bigg]\Bigg[\sum_{t=1}^T\psi_t\psi_t'\Bigg]^{-1},
\end{eqnarray}
where $\alpha_t^{(j)}=g_t-i_n-\big(\mathsf{E}_t^{(j)}\big)'-\text{diag}\Big\{\delta_{t|T}^{(j)}\Big\}\Omega_{uv}v_t$ is an $(n\times 1)$ vector. If $i$--th company does not pay a dividend at time $t$, then as mentioned before $g_{i,t}=1$. In this case, because $i$--th components of the vectors $\delta_t$ and $\delta_{t|T}^{(j)}$ are zero, $i$--th component of the vector $\alpha_t^{(j)}$ also equals zero. Consequently, if all companies do not pay dividends, then the ML equations of parameters $\mu_0$, $C_k$, and $C_m$ are obtained by setting $\alpha_t^{(j)}=0$ in ML equations \eqref{03145}--\eqref{03147}.

To obtain estimations of the parameter matrices $C_z$ and $A$, let us define the following stacked vector and matrices: $\bar{z}_{t-1}^*:=(\psi_t',(z_{t-1}^*)')'$ is an $([l+\ell p]\times 1)$ vector, $\bar{A}:=[C_z:A]$ is an $(\ell\times [l+\ell p]])$ matrix, $U:=[u_1:\dots:u_T]$ is an $(n\times T)$ matrix, $V:=[v_1:\dots:v_T]$ is an $(\ell\times T)$ matrix, $\mathsf{Z}:=[z_1:\dots:z_T]$ is an $(\ell\times T)$ matrix, and $\bar{\mathsf{Z}}_{-1}^*:=[\bar{z}_0^*:\dots:\bar{z}_{T-1}^*]$ is an $([l+\ell p]\times T)$ matrix. Then, the second line of the system \eqref{03017} can be written by
\begin{equation}\label{03148}
z_t=\bar{A}\bar{z}_{t-1}^*+v_t, ~~~t=1,\dots,T.
\end{equation}  
Consequently, the above equation is compactly written by
\begin{equation}\label{03149}
\mathsf{Z}=\bar{A}\bar{\mathsf{Z}}_{-1}^*+V.
\end{equation}  
Since $\sum_{t=1}^Tv_tu_t'=VU'$ and $\sum_{t=1}^Tv_tv_t'=VV'$, and $V$ is measurable with respect to the full information $\mathcal{F}_T$, one obtains that 
\begin{equation}\label{03150}
\sum_{t=1}^T\mathbb{E}\Big[u_t'\Omega_{uv}v_t\Big|\mathcal{F}_T;\theta^{(j)}\Big]=\text{tr}\Big\{\big(U_{\bullet|T}^{(j)}\big)'\Omega_{uv}V\Big\}~\text{and}~\sum_{t=1}^T\mathbb{E}\Big[v_t'\Omega_{vv}v_t\Big|\mathcal{F}_T;\theta^{(j)}\Big]=\text{tr}\Big\{V'\Omega_{vv}V\Big\},
\end{equation}
where $U_{\bullet|T}^{(j)}:=\big[u_{1|T}^{(j)}:\dots:u_{T|T}^{(j)}\big]$ is an $(n\times T)$ smoothed matrix, corresponding to the white noise process $u_t$. Therefore, due to \citeA{Lutkepohl05}, partial derivatives of the above equations with respect to the matrix $\bar{A}$ are given by 
\begin{equation}\label{03151}
\frac{\partial}{\partial \bar{A}}\sum_{t=1}^T\mathbb{E}\Big[u_t'\Omega_{uv}v_t\Big|\mathcal{F}_T;\theta^{(j)}\Big]=-\Omega_{vu}U_{\bullet|T}^{(j)}\big(\bar{\mathsf{Z}}_{-1}^*\big)'
\end{equation}
and
\begin{equation}\label{03152}
\frac{\partial}{\partial \bar{A}}\sum_{t=1}^T\mathbb{E}\Big[v_t'\Omega_{vv}v_t\Big|\mathcal{F}_T;\theta^{(j)}\Big]=2\Omega_{vv}\bar{A}\bar{\mathsf{Z}}_{-1}^*\big(\bar{\mathsf{Z}}_{-1}^*\big)'-2\Omega_{vv}\mathsf{Z}\big(\bar{\mathsf{Z}}_{-1}^*\big)'.
\end{equation}
Consequently, a partial derivative of the log--likelihood function $\Lambda(\theta|\mathcal{F}_T;\theta^{(j)})$ with respect to the matrix $\bar{A}$ is given by
\begin{equation}\label{03153}
\frac{\Lambda(\theta|\mathcal{F}_T;\theta^{(j)})}{\partial \bar{A}}=\Omega_{vu}U_{\bullet|T}^{(j)}\big(\bar{\mathsf{Z}}_{-1}^*\big)'+\Omega_{vv}\mathsf{Z}\big(\bar{\mathsf{Z}}_{-1}^*\big)'-\Omega_{vv}\bar{A}\bar{\mathsf{Z}}_{-1}^*\big(\bar{\mathsf{Z}}_{-1}^*\big)'
\end{equation}
As a result, if we equate the above equation to zero, then one obtains an estimator of the parameter matrix $\bar{A}$
\begin{equation}\label{03154}
\bar{A}:=\Big(\Omega_{vv}^{-1}\Omega_{vu}U_{\bullet|T}^{(j)}\big(\bar{\mathsf{Z}}_{-1}^*\big)'+\mathsf{Z}\big(\bar{\mathsf{Z}}_{-1}^*\big)'\Big)\big(\bar{\mathsf{Z}}_{-1}^*\big(\bar{\mathsf{Z}}_{-1}^*\big)'\big)^{-1}.
\end{equation}

For estimators of the covariance matrices $\Sigma_{\eta\eta}$, $\Sigma_{ww}$, and $\Sigma_0$, the following formulas holds
\begin{equation}\label{03155}
\Sigma_{\eta\eta}:=\frac{1}{T}\sum_{t=1}^T\mathbb{E}\big[\eta_t\eta_t'\big|\mathcal{F}_T;\theta^{(j)}\big]=\frac{1}{T}\sum_{t=1}^T\begin{bmatrix}
\mathbb{E}\big[u_tu_t'\big|\mathcal{F}_T;\theta^{(j)}\big] & \mathbb{E}\big[u_tv_t'\big|\mathcal{F}_T;\theta^{(j)}\big]\\
\mathbb{E}\big[v_tu_t'\big|\mathcal{F}_T;\theta^{(j)}\big] & \mathbb{E}\big[v_tv_t'\big|\mathcal{F}_T;\theta^{(j)}\big]
\end{bmatrix},
\end{equation}
\begin{equation}\label{03156}
\Sigma_{ww}:=\frac{1}{T}\sum_{t=1}^T\mathbb{E}\big[w_tw_t'\big|\mathcal{F}_T;\theta^{(j)}\big], 
~~~~\text{and}~~~\Sigma_0:=\Sigma^{(j)}(\tilde{m}_0|T).
\end{equation}
To calculate the conditional expectations $\mathbb{E}\big[\eta_t\eta_t'|\mathcal{F}_T;\theta^{(j)}\big]$ and $\mathbb{E}\big[w_tw_t'|\mathcal{F}_T;\theta^{(j)}\big]$, observe that the white noise processes at time $t$ of the log book value growth rate process, the economic variables process, and the log price--to--book ratio process can be represented by 
\begin{eqnarray}\label{03157}
u_t&=&u_{t|T}+\mathsf{R}_t\big(\tilde{m}_t^*-\tilde{m}_{t|T}^*\big)\nonumber\\
v_t&=&z_t-\bar{A}\bar{z}_{t-1}^*\\
w_t&=&w_{t|T}+\mathsf{R}\big(\tilde{m}_t^*-\tilde{m}_{t|T}^*\big)\nonumber.
\end{eqnarray}
Therefore, as $v_t$, $u_{t|T}$, and $v_{t|T}$ are measurable with respect to the full information $\mathcal{F}_T$ (their values are known at time $T$), it follows from equations \eqref{03157} that 
\begin{equation}\label{03158}
\mathbb{E}\big[u_tu_t'|\mathcal{F}_T;\theta^{(j)}\big]=u_{t|T}^{(j)}\big(u_{t|T}^{(j)}\big)'+\mathsf{R}_t\Sigma^{(j)}(\tilde{m}_t^*|T)\mathsf{R}_t',~~~
\mathbb{E}\big[u_tv_t'|\mathcal{F}_T;\theta^{(j)}\big]=u_{t|T}^{(j)}v_t',
\end{equation}
\begin{equation}\label{03159}
\mathbb{E}\big[v_tv_t'|\mathcal{F}_T;\theta^{(j)}\big]=v_tv_t',~~~\text{and}~~~\mathbb{E}\big[w_tw_t'|\mathcal{F}_T;\theta^{(j)}\big]=w_{t|T}^{(j)}\big(w_{t|T}^{(j)}\big)'+\mathsf{R}\Sigma^{(j)}(\tilde{m}_t^*|T)\mathsf{R}'.
\end{equation}

To get estimations at iteration $(j+1)$ of the parameters $\Sigma_{\eta\eta}$ and $\Sigma_{ww}$, we need to substitute equations \eqref{03158} and \eqref{03159} into equations \eqref{03155} and \eqref{03156}.
After substitutions, to obtain estimations at iteration $(j+1)$ of the model's parameters, we needs to solve ML equations \eqref{03145}--\eqref{03147} and \eqref{03154}--\eqref{03156} for the parameters applying the Kalman smoothing at iteration $j$. As a result, then under suitable conditions the zig--zag iteration that corresponds to equations \eqref{03121}--\eqref{03128}, \eqref{03133}, \eqref{03134}, \eqref{03145}--\eqref{03147}, and \eqref{03154}--\eqref{03156} converges to the maximum likelihood estimation of our log private company valuation model, corresponding to equation \eqref{03010}. 

Now, we consider the case, where the dividend proportional to market price of equity, see equation \eqref{ad001}. Then, it can be shown that ML equations of the parameters $\mu_0$, $C_k$, and $C_m$ are given by the following equations
\begin{equation}\label{ad002}
\mu_0:=\tilde{m}_{0|T}^{(j)},
\end{equation}
\begin{eqnarray}\label{ad003}
C_k:=\Omega_{uu}^{-1}\Bigg[\sum_{t=1}^T\bigg(\Omega_{uu}u_{t|T}^{k(j)}-\alpha_t^{(j)}-\Omega_{uv}v_t\bigg)\psi_t'\Bigg]\Bigg[\sum_{t=1}^T\psi_t\psi_t'\Bigg]^{-1},
\end{eqnarray}
and
\begin{eqnarray}\label{ad004}
C_m:=\Bigg[\sum_{t=1}^T\bigg\{\mathsf{R}\tilde{m}_{t|T}^{*(j)}\psi_t'\bigg\}\Bigg]\Bigg[\sum_{t=1}^T\psi_t\psi_t'\Bigg]^{-1},
\end{eqnarray}
where $\alpha_t^{(j)}=g_t-i_n-\big(\mathsf{E}_t^{(j)}\big)'-\text{diag}\Big\{\delta_{t|T}^{(j)}\Big\}\Omega_{uv}v_t$ is an $(n\times 1)$ vector with $\delta_{t|T}^{(j)}:=(G_t-I_n)\big(u_{t|T}^{(j)}-\tilde{m}_{t-1|T}^{(j)}\big)$ and representations \eqref{ad033} and \eqref{ad034} still hold for $\varphi:=\tilde{\Delta}_t-C_k\psi_t$. ML equations of the other parameters are the same as the previous case, where dividend is proportional to book value. Again, under the suitable conditions the zig--zag iteration converges to the ML estimations of the parameters.

A smoothed inference of the market value vector at time $t$ of the private companies is calculated by the following formula
\begin{equation}\label{03160}
V_{t|T}=m_{t|T}\odot B_t, ~~~t=0,1,\dots,T,
\end{equation}
where $m_{t|T}=\exp\{\tilde{m}_{t|T}\}$ is a smoothed price--to--book ratio vector at time $t$. Also, an analyst can forecast the market value process of the private companies by using equation \eqref{03129}. 

\section{Conclusion}

For a public company, option pricing models and life insurance models, relying on the observed price of the company have been developed. However, for a private company, because of unobserved prices, pricing and hedging of options and life insurance products are in their early stages of development. For this reason, this paper introduces a log private company valuation model, which is based on the dynamic Gordon growth model. To the best of our knowledge, this is the first attempt to introduce an option pricing model for a private company. Since most private companies pay dividends, in this paper, we consider the dividend--paying Black--Scholes call and put options and obtain closed--form pricing formulas, which are based on book values, log book value growth rates, and unobserved log price--to--book ratios of a private company for the options. Next, we obtain closed--form pricing formulas for the equity--linked life insurance products, which consist of term life and pure endowment insurances, corresponding to segregated fund contracts and unit--linked life insurance. Because hedging is an important concept for both options and life insurance products, we derived locally risk--minimizing strategies for the Black--Scholes call and put options and the equity--linked life insurance products. Finally, in order to use the model, we provide ML estimations of the model's parameters and EM algorithm, which are focused on the Kalman filtering. 

\section{Technical Annex}

Here we give the Lemmas and their proofs.

\begin{lem}\label{lem01}
Let $X\sim \mathcal{N}(\mu,\sigma^2)$. Then for all $K>0$,
\begin{equation*}\label{03174}
\mathbb{E}\big[\big(e^X-K\big)^+\big]=\exp\bigg\{\mu+\frac{\sigma^2}{2}\bigg\}\Phi(d_1)-K\Phi(d_2)
\end{equation*}
and
\begin{equation*}\label{03174}
\mathbb{E}\big[\big(K-e^X\big)^+\big]=K\Phi(-d_2)-\exp\bigg\{\mu+\frac{\sigma^2}{2}\bigg\}\Phi(-d_1),
\end{equation*}
where $d_1:=\big(\mu+\sigma^2-\ln(K)\big)/\sigma$, $d_2:=d_1-\sigma$, and $\Phi(x)=\int_{-\infty}^x\frac{1}{\sqrt{2\pi}}e^{-u^2/2}du$ is the cumulative standard normal distribution function.
\end{lem}
\begin{proof}
See, e.g., \citeA{Battulga24a}.
\end{proof}

\begin{lem}\label{lem02}
Let $\alpha_1\in \mathbb{R}^{n_1}$ and $\alpha_2\in\mathbb{R}^{n_2}$ be fixed vectors, and $X_1\in \mathbb{R}^{n_1}$ and $X_2\in \mathbb{R}^{n_2}$ be random vectors and their joint distribution is given by
$$\begin{bmatrix}
X_1 \\ X_2
\end{bmatrix} \sim \mathcal{N}\bigg(\begin{bmatrix}
\mu_1 \\ \mu_2
\end{bmatrix},\begin{bmatrix}
\Sigma_{11} & \Sigma_{12}\\
\Sigma_{21} & \Sigma_{22}
\end{bmatrix}\bigg).$$
Then, for all $L\in\mathbb{R}_{+}^{n_2}$, it holds
\begin{eqnarray*}
&&\Psi^+\big(L;\alpha_1;\alpha_2;\mu_1;\mu_2;\Sigma_{11};\Sigma_{12};\Sigma_{22}\big):=\mathbb{E}\bigg[\Big(\alpha_1\odot e^{X_1}\Big)\Big(\alpha_2\odot\big(e^{X_2}-L\big)^+\Big)'\bigg]\\
&&=\bigg(\Big(\alpha_1\odot\mathbb{E}\big[e^{X_1}\big]\Big)\Big(\alpha_2\odot\mathbb{E}\big[e^{X_2}\big]\Big)'\bigg)\odot e^{\Sigma_{12}}\odot\Phi\bigg(i_{n_1}\otimes d_1'+\Sigma_{12}\mathrm{diag}\big\{\mathcal{D}\big[\Sigma_{22}\big]\big\}^{-1/2}\bigg)\\
&&-\bigg(\Big(\alpha_1\odot\mathbb{E}\big[e^{X_1}\big]\Big)\Big(\alpha_2\odot L\Big)'\bigg)\odot\Phi\bigg(i_{n_1}\otimes d_2'+\Sigma_{12}\mathrm{diag}\big\{\mathcal{D}\big[\Sigma_{22}\big]\big\}^{-1/2}\bigg)
\end{eqnarray*}
and
\begin{eqnarray*}
&&\Psi^-\big(L;\alpha_1;\alpha_2;\mu_1;\mu_2;\Sigma_{11};\Sigma_{12};\Sigma_{22}\big):=\mathbb{E}\bigg[\Big(\alpha_1\odot e^{X_1}\Big)\Big(\alpha_2\odot\big(L-e^{X_2}\big)^+\Big)'\bigg]\\
&&=\bigg(\Big(\alpha_1\odot\mathbb{E}\big[e^{X_1}\big]\Big)\Big(\alpha_2\odot L\Big)'\bigg)\odot \Phi\bigg(-i_{n_1}\otimes d_2'-\Sigma_{12}\mathrm{diag}\big\{\mathcal{D}\big[\Sigma_{22}\big]\big\}^{-1/2}\bigg)\\
&&-\bigg(\Big(\alpha_1\odot\mathbb{E}\big[e^{X_1}\big]\Big)\Big(\alpha_2\odot \mathbb{E}\big[e^{X_2}\big]\Big)'\bigg)\odot e^{\Sigma_{12}}\odot\Phi\bigg(-i_{n_1}\otimes d_1'-\Sigma_{12}\mathrm{diag}\big\{\mathcal{D}\big[\Sigma_{22}\big]\big\}^{-1/2}\bigg),
\end{eqnarray*}
where for each $i=1,2$, $\mathbb{E}\big[e^{X_i}\big]=e^{\mu_i+1/2\mathcal{D}[\Sigma_{ii}]}$ is the expectation of the multivariate log--normal random vector, $d_1:=\big(\mu_2+\mathcal{D}[\Sigma_{22}]-\ln(L)\big)\oslash\sqrt{\mathcal{D}[\Sigma_{22}]}$, $d_2:=d_1-\sqrt{\mathcal{D}[\Sigma_{22}]}$, and $\Phi(x)=\int_{-\infty}^x\frac{1}{\sqrt{2\pi}}e^{-u^2/2}du$ is the cumulative standard normal distribution function.
\end{lem}
\begin{proof}
For $i=1,\dots,n_1$ and $j=1,\dots,n_2$, let us denote an $i$--th elements of the fixed vector $\alpha_1$ and random vector $X_1$ by $(\alpha_1)_i$ and $(X_1)_i$, respectively, and a $j$--th elements of the fixed vector $\alpha_2$ and random vector $X_2$ by $(\alpha_2)_j$ and $(X_2)_j$, respectively, a mean of the random variable $(X_2)_j$ by $(\mu_2)_j$, a variance of the random variable $(X_2)_j$ by $(\mathcal{D}[\Sigma_{22}])_j$, a covariance between the random vectors $(X_1)_i$ and $(X_2)_j$ by $(\Sigma_{12})_{ij}$, and a $j$--th element of the vector $L$ by $L_j$. Then, it follows from Lemma 5 in \citeA{Battulga24a}, that
\begin{eqnarray*}
&&\mathbb{E}\Big[\Big((\alpha_1)_ie^{(X_1)_i}\Big)\Big((\alpha_2)_j\big(e^{(X_2)_j}-K\big)^+\Big)\Big]=\Big((\alpha_1)_i\mathbb{E}\big[e^{(X_1)_i}\big]\Big)\Big((\alpha_2)_j\mathbb{E}\big[e^{(X_2)_j}\big]\Big)e^{(\Sigma_{12})_{ij}}\\
&&\times \Phi\Bigg((d_1)_j+\frac{(\Sigma_{12})_{ij}}{\sqrt{(\mathcal{D}[\Sigma_{22}])_j}}\Bigg)-\Big((\alpha_1)_i\mathbb{E}\big[e^{(X_1)_i}\big]\Big)\Big((\alpha_2)_jL_j\Big)\Phi\Bigg((d_2)_j+\frac{(\Sigma_{12})_{ij}}{\sqrt{(\mathcal{D}[\Sigma_{22}])_j}}\Bigg)
\end{eqnarray*}
and
\begin{eqnarray*}
&&\mathbb{E}\Big[\Big((\alpha_1)_ie^{(X_1)_i}\Big)\Big((\alpha_2)_j\big(K-e^{(X_2)_j}\big)^+\Big)\Big]=\Big((\alpha_1)_i\mathbb{E}\big[e^{(X_1)_i}\big]\Big)\Big((\alpha_2)_jL_j\Big)\Phi\Bigg(-(d_2)_j-\frac{(\Sigma_{12})_{ij}}{\sqrt{(\mathcal{D}[\Sigma_{22}])_j}}\Bigg)\\
&&-\Big((\alpha_1)_i\mathbb{E}\big[e^{(X_1)_i}\big]\Big)\Big((\alpha_2)_j\mathbb{E}\big[e^{(X_2)_j}\big]\Big)e^{(\Sigma_{12})_{ij}}\Phi\Bigg(-(d_1)_j-\frac{(\Sigma_{12})_{ij}}{\sqrt{(\mathcal{D}[\Sigma_{22}])_j}}\Bigg)
\end{eqnarray*}
for $i=1,\dots,n_1$ and $j=1,\dots,n_2$, where $(d_1)_j:=\big((\mu_2)_j+(\mathcal{D}[\Sigma_{22}])_j-\ln(L_j)\big)/\sqrt{(\mathcal{D}[\Sigma_{22}])_j}$ and $(d_2)_j:=(d_1)_j-\sqrt{(\mathcal{D}[\Sigma_{22}])_j}$. That completes the proof the Lemma.
\end{proof}

\begin{lem}\label{lem03}
Let $u_t=G_t^{-1}(\tilde{b}_t+\tilde{m}_t-\tilde{\Delta}_t)-\tilde{m}_{t-1}+\tilde{\Delta}_t-C_k\psi_t+G_t^{-1}h_t$ and $w_t=\tilde{m}_t-C_m\psi_t-\tilde{m}_{t-1}$. Then, the following partial derivatives hold
\begin{equation}\label{03161}
\frac{\partial u_t}{\partial \mu_0'}=\big[L_1\delta_t:\dots:L_n\delta_t\big],~~~\frac{\partial \ln(|G_t|)}{\partial \mu_0'}=-(g_t-i_n)'
\end{equation}
\begin{equation}\label{03162}
\frac{\partial u_t}{\partial c_k'}=\Big(\big[L_1\delta_t:\dots:L_n\delta_t\big]-I_n\Big)(\psi_t'\otimes I_n),~~~\frac{\partial \ln(|G_t|)}{\partial c_k'}=-(g_t-i_n)'(\psi_t'\otimes I_n),
\end{equation}
\begin{equation}\label{03163}
\frac{\partial u_t}{\partial c_m'}=\sum_{s=1}^{t-1}\big[L_1\delta_t:\dots:L_n\delta_t\big](\psi_s'\otimes I_n), ~~~\frac{\partial \ln(|G_t|)}{\partial c_m'}=-\sum_{s=1}^{t-1}(g_t-i_n)'(\psi_s'\otimes I_n),
\end{equation}
and
\begin{equation}\label{03164}
\frac{\partial }{\partial c_m'}\Big(w_t'\Sigma_{ww}^{-1}w_t\Big)=-2w_t'\Sigma_{ww}^{-1}(\psi_t'\otimes I_n),
\end{equation}
where $\delta_t$ and $L_i$ ($i=1,\dots,n$) are given by equations \eqref{03137} and \eqref{03138}.
\end{lem}
\begin{proof}
Observe that $\tilde{P}_{t-1}-\tilde{d}_t=\tilde{m}_{t-1}-\tilde{\Delta}_t$. Then, due to equation \eqref{03006}, the log dividend--to--price ratio is represented by $\tilde{d}_t-\tilde{P}_t=h_t-G_t(\tilde{m}_{t-1}-\tilde{\Delta}_t+C_k\psi_t)-G_tu_t.$ Therefore, as $\mu_t=\mathbb{E}[\tilde{d_t}-\tilde{P}_t|\mathcal{F}_0]$, we get that
\begin{equation}\label{03165}
\mu_t=h_t+G_t\big(\tilde{\Delta}_t-C_k\psi_t-\mathbb{E}[\tilde{m}_{t-1}|\mathcal{F}_0]\big).
\end{equation}
Consequently, because $h_t=G_t\big(\ln(g_t)-\mu_t\big)+\mu_t=G_t\big(\ln\big(g_t\oslash(g_t-i_n)\big)\big)+\mu_t$ and $\mathbb{E}[\tilde{m}_{t-1}|\mathcal{F}_0]=\mu_0+C_m\sum_{s=1}^{t-1}\psi_s$, one obtains that 
\begin{equation}\label{03166}
g_t=\exp\{-\varphi_t\}\oslash\big(\exp\{-\varphi_t\}-i_n\big),
\end{equation}
where $\varphi_t:=\tilde{\Delta}_t-C_k\psi_t-\big(\mu_0+C_m\sum_{s=1}^{t-1}\psi_s\big)$. Further, since $\mu_t=\ln(g_t-i_n)=-\ln\big(\exp\{-\varphi_t\}-i_n\big)$, we get that
\begin{equation}\label{03167}
h_t=-\Big\{\varphi_t\odot\exp\{-\varphi_t\}\oslash\Big(\exp\{-\varphi_t\}-i_n\Big)+\ln\Big(\exp\{-\varphi_t\}-i_n\Big)\Big\}.
\end{equation}
Thus, for $i=1,\dots,n$, partial derivatives of the linearization parameters $g_{i,t}$ and $h_{i,t}$ are obtained as 
\begin{equation}\label{03168}
\frac{\partial g_{i,t}}{\partial \alpha}=(g_{i,t}-1)g_{i,t}\frac{\partial \varphi_{i,t}}{\partial \alpha}~~~\text{and}~~~
\frac{\partial h_{i,t}}{\partial \alpha}=-\varphi_{i,t}(g_{i,t}-1)g_{i,t}\frac{\partial \varphi_{i,t}}{\partial \alpha},
\end{equation}
where $\alpha$ equals elements of the vector $\mu_0$ and the matrices $C_k$ and $C_m$.  Since
\begin{equation}\label{03169}
\frac{\partial}{\partial \alpha}\ln(|G_t|)=\sum_{i=1}^n(g_{i,t}-1)\frac{\partial \varphi_{i,t}}{\partial \alpha},
\end{equation}
\begin{equation}\label{03170}
\frac{\partial}{\partial \alpha}\bigg(\frac{1}{g_{i,t}}\bigg)=-\frac{g_{i,t}-1}{g_{i,t}}\frac{\partial \varphi_{i,t}}{\partial \alpha},~~~\text{and}~~~\frac{\partial}{\partial \alpha}\bigg(\frac{h_{i,t}}{g_{i,t}}\bigg)=-\Big(\varphi_{i,t}g_{i,t}+h_{i,t}\Big)\frac{g_{i,t}-1}{g_{i,t}}\frac{\partial \varphi_{i,t}}{\partial \alpha},
\end{equation}
a partial derivative of the white noise process $u_{i,t}$ with respect to the parameter $\alpha$ is obtained by the following equation 
\begin{equation}\label{03171}
\frac{\partial u_{i,t}}{\partial \alpha}=-\big(g_{i,t}-1\big)\Bigg(u_{i,t}+\tilde{m}_{i,t-1}-\bigg(\mu_0+C_m\sum_{s=1}^{t-1}\psi_s\bigg)_i\Bigg)\frac{\partial \varphi_{i,t}}{\partial \alpha}-\frac{\partial (C_k\psi_t)_i}{\partial \alpha}, 
\end{equation}
where we use the fact that $\ln(|G_t|)=\sum_{i=1}^n\ln(|g_{i,t}|)$. Therefore, in vector form, equations \eqref{03169} and \eqref{03171} can be written by
\begin{equation}\label{03172}
\frac{\partial}{\partial \alpha}\ln(|G_t|)=(g_t-i_n)'\frac{\partial \varphi_t}{\partial \alpha}
\end{equation}
and
\begin{equation}\label{03173}
\frac{\partial u_t}{\partial \alpha}=-\text{diag}\{\delta_t\}\frac{\partial \varphi_t}{\partial \alpha}-\frac{\partial (C_k\psi_t)}{\partial \alpha}=-\big[L_1\delta_t:\dots:L_n\delta_t\big]\frac{\partial \varphi_t}{\partial \alpha}-\frac{\partial (C_k\psi_t)}{\partial \alpha},
\end{equation}
respectively.
Let us denote vectorizations of the matrices $C_k$ and $C_m$ by $c_k:=\text{vec}(C_k)$ and $c_m:=\text{vec}(C_m)$, respectively. Consequently, since $\frac{\partial (C_k\psi_t)}{\partial c_k'}=\frac{\partial ((\psi_t'\otimes I_n)c_k)}{\partial c_k'}=(\psi_t'\otimes I_n)$ and $\frac{\partial (C_m\psi_s)}{\partial c_m'}=(\psi_s'\otimes I_n)$, partial derivatives of the white noise process $u_t$ and $\ln(|G_t|)$ with respect to the parameters $\mu_0$, $c_k$, and $c_m$ are given by the equations \eqref{03161}--\eqref{03163}. Also, as $\frac{\partial (C_m\psi_t)}{\partial c_m'}=\psi_t'\otimes I_n$, one obtains equation \eqref{03164}. That completes the proof.
\end{proof}

\bibliographystyle{apacite}
\bibliography{References}

\end{document}